\documentclass[journal]{IEEEtran}

\usepackage{amsmath,bm}
\usepackage{relsize}
\usepackage[dvips]{graphicx}
\usepackage{algorithm}
\usepackage{algorithmicx}
\usepackage{algpseudocode}
\usepackage{amsfonts}
\usepackage{cite}
\usepackage{dsfont}
\usepackage{amssymb}
\usepackage{color}
\usepackage{tabls}
\usepackage{slashbox}
\usepackage{stfloats}
\usepackage[table,xcdraw]{xcolor}
\usepackage{hhline}
\usepackage{tabu}
\usepackage[caption=false]{subfig}
\usepackage[breaklinks]{hyperref}
\newtheorem{thm} {Theorem}
\newtheorem{asump} {Assumption}

\newtheorem{lemma} {Lemma}

\renewcommand\IEEEkeywordsname{Index Terms}
\let\oldIEEEkeywords\IEEEkeywords
\def\IEEEkeywords{\oldIEEEkeywords\normalfont\bfseries\ignorespaces}

\makeatletter
\newcommand{\thickhline}{
	\noalign {\ifnum 0=`}\fi \hrule height 1pt
	\futurelet \reserved@a \@xhline
}
\makeatother

\makeatletter
\newcounter{parentalgorithm}

\makeatother

\begin{document}
	
	\title{Information-Sharing over Adaptive Networks with Self-interested Agents}

	\author{Chung-Kai~Yu,~\IEEEmembership{Student Member,~IEEE,}
		Mihaela~van~der~Schaar,~\IEEEmembership{Fellow,~IEEE,}
		and~Ali~H.~Sayed,~\IEEEmembership{Fellow,~IEEE}
		\thanks{Copyright (c) 2015 IEEE. Personal use of this material is permitted. However, permission to use this material for any other	purposes must be obtained from the IEEE by sending a request to	pubs-permissions@ieee.org.}
		\thanks{This work was supported in part by NSF grants CCF-1011918, CSR-1016081, and ECCS-1407712. An early short version of this work appeared in the conference publication~\cite{Yu132}. The authors are with the Department of Electrical Engineering, University of California, Los Angeles, CA 90095-1594 USA (e-mail: ckyuna@ucla.edu, \{mihaela,sayed\}@ee.ucla.edu).}}
	
	\maketitle
	
	\begin{abstract}
		We examine the behavior of multi-agent networks where information-sharing is subject to a positive
		communications cost over the edges linking the agents. We consider a general mean-square-error
		formulation where all agents are interested in estimating the same target vector. We first show that, in the absence of any incentives to cooperate, the optimal strategy for the agents is to behave in a selfish manner with each agent seeking the optimal solution independently of the other agents. Pareto inefficiency arises as a result of the fact that agents are not using historical data to predict the behavior of their neighbors and to know whether they will reciprocate and participate in sharing information. Motivated by this observation, we develop a reputation protocol to summarize the opponent's past actions into a reputation score, which can then be used to form a belief about the opponent's subsequent actions. The reputation protocol entices agents to cooperate and turns their optimal strategy into an action-choosing strategy that enhances the overall social benefit of the network. In particular, we show that when the communications cost becomes
		large, the expected social benefit of the proposed protocol outperforms the social benefit that is obtained by cooperative agents that always share data. We perform a detailed mean-square-error analysis of the evolution of the network over three domains: far field, near-field, and middle-field, and show that the network behavior is stable for sufficiently small step-sizes. The various theoretical results are illustrated by numerical simulations.
	\end{abstract}
	
	
	\IEEEpeerreviewmaketitle
	
	\section{Introduction}
	
	\IEEEPARstart{A}{daptive} networks enable agents to share information and to solve distributed optimization and inference tasks in an efficient and decentralized manner. In most prior works, agents are assumed to be cooperative and designed to follow certain distributed rules such as the consensus strategy (e.g.,~\cite{Braca08,Dimakis10,Sardellitti10,Olfati07,Kar09,Kar11,Nedic09,Xiao04,Boyd06}) or the diffusion strategy (e.g.,~\cite{Lopes08,Sayed14,Sayed132,Sayed142,Sayed143,Catt10,Chouvardas11,Takahashi10,Chen12}). These rules generally include a self-learning step to update the agents' estimates using their local data, and a social-learning step to fuse and combine the estimates shared by neighboring agents. 
	However, when agents are selfish, they would not obey the preset rules unless these strategies conform to their own interests, such as minimizing their own costs. 
	In this work, we assume that the agents can behave selfishly and that they, therefore, have the freedom to decide whether or not they want to participate in sharing information with their neighbors at every point in time. Under these conditions, the global social benefit for the network can be degraded unless a policy is introduced to entice agents to participate in the collaborative process despite their individual interests. In this article, we will address this difficulty in the context of adaptive networks where agents are continually subjected to streaming data, and where they can predict in real-time, from their successive interactions, how reliable their neighbors are and whether they can be trusted to share information based on their past history. This formulation is different from the useful work in~\cite{Gharehshiran13}, which considered one particular form of selfish behavior in the context of a game-theoretic formulation. In that work, the focus is on activating the self-learning and social learning steps simultaneously, and agents simply decide whether to enter into a sleep mode (to save energy) or to continue acquiring and processing data. In the framework considered in our work, agents always remain active and are continually acquiring data; the main question instead is to entice agents to participate in the collaborative information-sharing process regardless of their self-centered evaluations.  

	More specifically, we study the behavior of multi-agent networks where information-sharing is subject to a positive communication cost over the edges linking the agents. This situation is common in applications, such as information sharing over cognitive networks~\cite{Yu13}, online learning under communication bandwidth and/or latency constraints ~\cite{Filali09},\cite[Ch.~14]{Bekkerman11}, and over social learning networks when the delivery of opinions involves some costs such as messaging fees~\cite{Acemoglu11,Jadbabaie12,Zhao12C2}. In our network model, each agent is self-interested and seeks to minimize its own sharing cost \textit{and} estimation error. Motivated by the practical scenario studied in~\cite{Yu13}, we formulate a general mean-square error estimation problem where all agents are interested in estimating the same target parameter vector. Agents are assumed to be foresighted and to have bounded rationality~\cite{Gigerenzer02} in the manner defined further ahead in the article. Then, we show that if left unattended, the dominant strategy for all agents is for them not to participate in the sharing of information, which leads to networks operating under an inefficient Pareto condition. This situation arises because agents do not have enough information to tell beforehand if their paired neighbors will reciprocate their actions (i.e., if an agent shares data with a second agent, will the second agent reciprocate and share data back?) This prediction-deficiency problem follows from the fact that agents are not using historical data to predict other agents' actions.
	
	One method to deal with this inefficient scenario is to assume that agents adapt to their opponents' strategies and improve returns by forming some regret measures. In~\cite{Foster99}, a decision maker determines its action using a regret measure to evaluate the utility loss from the chosen action to the optimal action in the previous stage game. For multi-agent networks, a regret-based algorithm was proposed in~\cite{Gharehshiran13} and~\cite{Namvar13} for agents to update their actions based on a weighted loss of the utility functions from the previous stage games. However, these works assume myopic agents and formulate repeated games with fixed utility functions over each stage game, which is different from the scenario considered in this article where the benefit of sharing information over adaptive networks continually evolves over time. This is because, as the estimation accuracy improves and/or the communication cost becomes expensive, the return to continue cooperating for estimation purposes falls and thus the act of cooperating with other agents becomes unattractive and inefficient. In this case, the regret measures computed from the previous stage games may not provide an accurate reference to the current stage game.  
	
	A second useful method to deal with Pareto inefficient and non-cooperative scenarios is to employ reputation schemes (e.g.,~\cite{Mailath06,Xu12,Zhang12,Carter02}). In this method, foresighted agents use reputation scores to assess the willingness of other agents to cooperate; the scores are also used to punish non-cooperative behavior. For example, the works~\cite{Xu12,Zhang12} rely on discrete-value reputation scores, say, on a scale 1-10, and these scores are regularly updated according to the agents' actions. Similar to the regret learning references mentioned before, in our problem the utilities or cost functions of stage games change over time and evolve based on agents' estimates. Conventional reputation designs do not address this time variation within the payoff of agents, which will be examined more closely in our work. Motivated by these considerations, in Sec.~\ref{AdRe}, we propose a dynamic/adaptive reputation protocol that is based on the belief measure of future actions with real-time benefit predictions. 
	
	In our formulation, we assume a general random-pairing model similar to~\cite{Boyd06}, where agents are randomly paired at the beginning of each time interval. This situation could occur, for example, due to an exogenous matcher or the mobility of the agents. The paired agents are assumed to follow a diffusion strategy\cite{Sayed14,Sayed132,Sayed142,Sayed143}, which includes an adaptation step and a consultation step, to iteratively update their estimates. Different from conventional diffusion strategies, the consultation step here is influenced by the random-pairing environment and by cooperation uncertainty. The interactions among self-interested agents are formulated as successive stage games of two players using pure strategies. To motivate agents to cooperate with each other, we formulate an adaptive reputation protocol to help agents jointly assess the instantaneous benefit of depreciating information and the transmission cost of sharing information. The reputation score helps agents to form a belief of their opponent's subsequent actions. Based on this belief, we entice agents to cooperate and turn their best response strategy into an action choosing strategy that conforms to Pareto efficiency and enhances the overall social benefit of the network. 
	
	In the performance evaluation, we are interested in ensuring the mean-square-error stability of the network instead of examining equilibria as is common in the game theoretical literature since our emphasis is on adaptation under successive time-variant stage games. The performance analysis is challenging due to the adaptive behavior by the agents. For this reason, we pursue the mean-square-error analysis of the evolution of the network over three domains: far-field, near-field, and middle-field, and show that the network behavior is stable for sufficiently small step-sizes. We also show that when information sharing becomes costly, the expected social benefit of the proposed reputation protocol outperforms the social benefit that is obtained by cooperative agents that always share data.
	
	\textbf{Notation}: We use lowercase letters to denote vectors and scalars, uppercase letters for matrices, plain letters for deterministic variables, and boldface letters for random variables. All vectors in our treatment are column vectors, with the exception of the regression vectors, $\bm{u}_{k,i}$.

	\section{System Model}
	
	\subsection{Distributed Optimization and Communication Cost}
	
	Consider a connected network consisting of $N$ agents. When agents act independently of each other, each agent $k$ would seek to estimate the $M \times 1$ vector $w^o$ that minimizes an individual estimation cost function denoted by $J_k^\text{est}(w):\mathbb{C}^M \rightarrow \mathbb{R}$. We assume each of the costs $\{J_k^\text{est}(w)\}$ is strongly convex for $k=1,2,\ldots,N$, and that all agents have the same objective so that all costs are minimized at the common location $w^o \in \mathbb{C}^{M \times 1}$. 
	
	In this work, we are interested in scenarios where agents can be motivated to cooperate among themselves as permitted by the network topology. We associate an extended cost function with each agent $k$, and denote it by $J_k(w,a_k)$. In this new cost, the scalar $a_k$ is a binary variable that is used to model whether agent $k$ is willing to cooperate and share information with its neighbors. The value $a_k=1$ means that agent $k$ is willing to share information (e.g., its estimate of $w^o$) with its neighbors, while the value $a_k=0$ means that agent $k$ is not willing to share information. The reason why agents may or may not share information is because this decision will generally entail some cost. We consider the scenario where a positive transmission cost, $c_k>0$, is required for each act by agent $k$ involving sharing an estimate with any of its neighbors. By taking $c_k$ into consideration, the extended cost $J_k(w,a)$ that is now associated with agent $k$ will consist of the sum of two components: the estimation cost and the communication cost\footnote{We focus on the sum of the estimation cost and the communication cost due to its simplicity and meaningfulness in applications. Note that a possible generalization is to consider a penalty-based objective function $J^\text{est}_k(w)+p(J^\text{com}_k(a_k))$ for some penalty function $p(\cdot)$.}:
	\begin{align}
		\label{Jcoop}
		J_k (w,a_k) \triangleq J^\text{est}_k(w) + J^\text{com}_k(a_k)
	\end{align}
	where the latter component is modeled as
	\begin{align}
		\label{comc}
		J^\text{com}_k(a_k)\triangleq a_k c_k
	\end{align}
	We express the communication expense in the form (\ref{comc}) because, as described further ahead, when an agent $k$ decides to share information, it will be sharing the information with one neighbor at a time; the cost for this communication will be $a_k c_k$. With regards to the estimation cost, $J_{k}^\text{est}(w)$, this measure can be selected in many ways. One common choice is the mean-square-error (MSE) cost, which we adopt in this work. 
	
	At each time instant $i\geq 0$, each agent $k$ is assumed to have access to a scalar measurement $\bm{d}_k(i) \in \mathbb{C}$ and a $1 \times M$ regression vector $\bm{u}_{k,i} \in \mathbb{C}^{1 \times M}$ with covariance matrix $R_{u,k} \triangleq \mathbb{E}\bm{u}_{k,i}^* \bm{u}_{k,i}>0$. 
	The regressors $\{\bm{u}_{k,i}\}$ are assumed to have zero-mean and to be temporally white and spatially independent, i.e., 
	\begin{align}
		\mathbb{E} \bm{u}_{k,i}^* \bm{u}_{\ell,j}=R_{u,k} \delta_{k \ell} \delta_{i j}
	\end{align}
	in terms of the Kronecker delta function. The data $\{\bm{d}_k(i),\bm{u}_{k,i}\}$ are assumed to be related via the linear regression model:
	\begin{align}
		\label{observe}
		\bm{d}_k(i)=\bm{u}_{k,i} w^o+\bm{v}_k(i)
	\end{align}
	where $w^o$ is the common target vector to be estimated by the agents. In (\ref{observe}), the variable $\bm{v}_k(i) \in \mathbb{C}$ is a zero-mean white-noise process with power $\sigma_{v,k}^2$ that is assumed to be spatially independent, i.e.,
	\begin{align}
		\mathbb{E}\bm{v}_k^*(i)\bm{v}_\ell(j)=\sigma_{v,k}^2\delta_{k\ell}\delta_{ij}
	\end{align}
	We further assume that the random processes $\bm{u}_{k,i}$ and $\bm{v}_\ell(i)$ are spatially and temporally independent for any $k$, $\ell$, $i$, and $j$. Models of the form (\ref{observe}) are common in many applications, e.g., channel estimation, model fitting, target tracking, etc (see, e.g.,~\cite{Sayed143}).
	
	Let $\bm{w}_{k,i-1}$ denote the estimator for $w^o$ that will be available to agent $k$ at time $i-1$. We will describe in the sequel how agents evaluate these estimates. The corresponding {\em a-priori}
	estimation error is defined by
	\begin{align}
		\label{apriori}
		\bm{e}_{a,k}(i) \triangleq \bm{d}_k(i)-\bm{u}_{k,i}\bm{w}_{k,i-1}
	\end{align}
	and it measures how close the weight estimate matches the measurements $\{\bm{d}_k(i),\bm{u}_{k,i}\}$ to each other. In view of model (\ref{observe}), we can also write
	\begin{align}
		\bm{e}_{a,k}(i)= \bm{u}_{k,i}\widetilde{\bm{w}}_{k,i-1}+\bm{v}_k(i)
	\end{align}
	in terms of the estimation error vector
	\begin{align}
		\widetilde{\bm{w}}_{k,i-1} \triangleq w^o-\bm{w}_{k,i-1}
	\end{align}
	Motivated by these expressions and model (\ref{observe}), the instantaneous MSE cost that is
	associated with agent $k$ based on the estimate from time $i-1$ is given by
	\begin{align}
		\label{qualcost1}
		J_k^\text{est}(\bm{w}_{k,i-1}) &\triangleq \mathbb{E}|\bm{e}_{a,k}(i)|^2 \notag\\
		&=\mathbb{E}|\bm{d}_k(i)-\bm{u}_{k,i}\bm{w}_{k,i-1}|^2 \notag\\
		&=\mathbb{E}\|\widetilde{\bm{w}}_{k,i-1}\|^2_{R_{u,k}}\;+\;\sigma_{v,k}^2
	\end{align}
	Note that this MSE cost conforms to the strong convexity of $J_k^\text{est}$ as we mentioned before.
	Combined with the action by agent $k$, the extended instantaneous cost at agent $k$ that is based on the prior estimate, $\bm{w}_{k,i-1}$, is then given by:
	\begin{align}
		\label{qualcost2}
		J_k(\bm{w}_{k,i-1},a_k) = \mathbb{E}|\bm{e}_{a,k}(i)|^2 + a_k c_k
	\end{align}
	
	\subsection{Random-Pairing Model}
	\label{RPM}
	We denote by $\mathcal{N}_k$ the neighborhood of each agent $k$, including itself. We consider a random pairing protocol for agents to share information at the beginning of every iteration cycle. The pairing procedure can be executed either in a centralized or distributed manner. Centralized pairing schemes can be used when an online server randomly assigns its clients into pairs as in crowdsourcing applications~\cite{Zhang12,Xu12}, or when a base-station makes pairing decisions for its mobile nodes for packet relaying~\cite{Yang09}. Distributed paring schemes arise more naturally in the context of economic and market transactions~\cite{Aliprantis07}. In our formulation, we adopt a distributed pairing structure that takes neighborhoods into account when selecting pairs, as explained next.
	
	We assume each agent $k$ has bi-directional links to other agents in $\mathcal{N}_k$ and that agent $k$ has a positive probability to be paired with any of its neighbors. Once two agents are paired, they can decide on whether to share or not their instantaneous estimates for $w^o$. We therefore model the result of the random-pairing process between each pair of agents $k$ and $\ell\in \mathcal{N}_k \setminus \{k\}$ as temporally-independent Bernoulli random processes defined as:
	\begin{align}
		\label{randompair}
		\mathbf{1}_{k \ell}(i)= \mathbf{1}_{\ell k }(i)=
		\begin{cases}
			1, & \text{with probability}~ p_{k \ell}=p_{\ell k}\\
			0, & \text{otherwise}
		\end{cases}
	\end{align}
	where $\mathbf{1}_{k \ell}(i)=1$ indicates that agents $k$ and $\ell$ are paired at time $i$ and $\mathbf{1}_{k \ell}(i)=0$ indicates that they are not paired. 
	We are setting $\mathbf{1}_{k \ell}(i)=\mathbf{1}_{\ell k}(i)$ because these variables represent the same event: whether agents $k$ and $\ell$ are paired, which results in $p_{k\ell}=p_{\ell k}$. For $\ell \notin \mathcal{N}_k$, we have $\mathbf{1}_{k \ell}(i)=0$ since such pairs will never occur. For convenience, we use $\mathbf{1}_{k k}(i)$ to indicate the event that agent $k$ is not paired with any agent $\ell\in \mathcal{N}_k \setminus \{k\}$ at time $i$, which happens with probability $p_{kk}$. Since each agent will pair itself with at most one agent at a time from its neighborhood, the following properties are directly obtained from the random-pairing procedure:
	\begin{align}
		&\sum_{\ell \in \mathcal{N}_k} \mathbf{1}_{k \ell}(i) = 1,~~~\sum\limits_{\ell \in \mathcal{N}_k} p_{k \ell} = 1 \\
		&\mathbf{1}_{k \ell}(i) \mathbf{1}_{k q}(i)=0, ~~ \text{for }\ell \neq q
	\end{align}
	We assume that the random pairing indicators $\{\mathbf{1}_{k \ell}(i)\}$ for all $k$ and $\ell$ are independent of the random variables $\{\bm{u}_{k,t}\}$ and $\{\bm{v}_k(t)\}$ for any time $i$ and $t$.
	For example, a widely used setting in the literature is the fully-pairing network, which assumes a fully-connected network topology~\cite{Ellison94,Zhang12}, i.e., $\mathcal{N}_k=\mathcal{N}$ for every agent $k$, where $\mathcal{N}$ denotes the set of all agents. The size $N=|\mathcal{N}|$ is assumed to be even and every agent is uniformly paired with exactly one agent in the network. Therefore, we have $N/2$ pairs at each time instant and the random-pairing probability becomes
	\begin{align}
		\label{fully}
		p_{k \ell} = 
		\begin{cases}
			\frac{1}{N-1}, & \text{for~} \ell \neq k \\
			0, & \text{for~} \ell = k
		\end{cases}
	\end{align}
	We will not be assuming fully-connected networks or fully-paired protocols and will deal more generally with networks that can be sparsely connected. Later in Sec. IV we will demonstrate a simple random-pairing protocol which can be implemented in a fully distributed manner.
		
	\subsection{Diffusion Strategy}
	
	\setcounter{equation}{24}
	\begin{figure*}[!b]
		\vspace{-1.5mm}
		\normalsize
		\hrulefill
		\begin{align}
			\label{a_sel}
			J_{k,i}^\infty \left[a_{k \ell}(i),a_{\ell k}(i)|\bm{w}_{k,i-1}\right] &\triangleq 
			\sum_{t=i}^{\infty} \delta_k^{t-i} \mathbb{E} 
			\Big[
			J_{k}(\bm{a}_{k \ell}(t),\bm{a}_{\ell k}(t)) \Big{|}
			\bm{w}_{k,i-1}, \bm{a}_{k \ell}(i)=a_{k \ell}(i), \bm{a}_{\ell k}(i)=a_{\ell k}(i) \Big] \notag \\
			&=\sum_{t=i}^{\infty} \delta_k^{t-i} \mathbb{E} 
			\left[
			J_{k}^\text{act}(\bm{a}_{\ell k}(t))+ \bm{a}_{k \ell}(t) c_k\Big{|} 
			\bm{w}_{k,i-1}, \bm{a}_{k \ell}(i)=a_{k \ell}(i), \bm{a}_{\ell k}(i)=a_{\ell k}(i) \right]
		\end{align}
	\end{figure*}
	\setcounter{equation}{14}
	
	Conventional diffusion strategies assume that the agents are cooperative (or obedient) and continuously share information with their neighbors as necessary~\cite{Sayed132,Sayed143,Sayed14}. 
	In the adapt-then-combine (ATC) version of diffusion adaptation, each agent $k$ updates its estimate, $\bm{w}_{k,i}$,  according to the following relations:
	\begin{align}
		\label{ATC_A}
		\bm{\psi}_{k,i} &= \bm{w}_{k,i-1} + \mu \bm{u}^*_{k,i}[\bm{d}_k(i)-\bm{u}_{k,i}\bm{w}_{k,i-1}] \\
		\label{genATC2}
		\bm{w}_{k,i} &= \sum_{\ell \in \mathcal{N}_k} \alpha_{\ell k}\bm{\psi}_{\ell,i}
	\end{align}
	where $\mu>0$ is the step-size parameter of agent $k$, and the $\{\alpha_{\ell k},\;\ell\in{\cal N}_k\}$ are nonnegative combination coefficients that add up to one.
	In implementation (\ref{ATC_A})--(\ref{genATC2}), each agent $k$ computes an intermediate estimate $\bm{\psi}_{k,i}$ using its local data, and subsequently fuses the intermediate estimates from its neighbors. For the combination step (\ref{genATC2}), since agent $k$ is allowed to interact with only one of its neighbors, then we rewrite (\ref{genATC2}) in terms of a single coefficient $0\leq \alpha_k \leq 1$ as follows:
	\begin{align}
		\label{ranATC2}
		\bm{w}_{k,i} = \begin{cases} 
			\alpha_k \bm{\psi}_{k,i} + (1-\alpha_k) \bm{\psi}_{\ell,i},\!\!\!&\text{if}~\mathbf{1}_{k \ell} (i)=1~\text{for some}~ \ell \neq k \\
			\bm{\psi}_{k,i},\!\!\!&\text{otherwise}
		\end{cases}
	\end{align}
	We can capture both situations in (\ref{ranATC2}) in a single equation as follows:
	\begin{align}
		\label{ATC_C}
		\bm{w}_{k,i} &= \alpha_k \bm{\psi}_{k,i} + (1-\alpha_k)\sum_{\ell \in \mathcal{N}_k} \mathbf{1}_{k \ell} (i) \bm{\psi}_{\ell,i}
	\end{align}
	In formulation (\ref{ATC_A}) and (\ref{ATC_C}), it is assumed that once agents $k$ and $\ell$ are paired, they share information according to (\ref{ATC_C}). 
	
	Let us now incorporate an additional layer into the algorithm in order to model instances of selfish behavior. When agents behave in a selfish (strategic) manner, even when agents $k$ and $\ell$ are paired, each one of them may still decide (independently) to refuse to share information with the other agent for selfish reasons (for example, agent $k$ may decide that this cooperation will cost more than the benefit it will reap for the estimation task). To capture this behavior, we use the specific notation $\bm{a}_{k\ell}(i)$, instead of $\bm{a}_k(i)$, to represent the action taken by agent $k$ on agent $\ell$ at time $i$, and similarly for $\bm{a}_{\ell k}(i)$. Both agents will end up sharing information with each other only if $\bm{a}_{k\ell}(i)=\bm{a}_{\ell k}(i)=1$, i.e., only when both agents are in favor of cooperating once they have been paired. We set $\bm{a}_{k k}(i)=1$ for every time $i$. We can now rewrite the combination step (\ref{ATC_C}) more generally as:
	\begin{align}
		\label{ATC2} \bm{w}_{k,i} = \alpha_k \bm{\psi}_{k,i} + (1-\alpha_k) \sum_{\ell \in \mathcal{N}_k} &\mathbf{1}_{k \ell} (i) [ \bm{a}_{\ell k} (i) \bm{\psi}_{\ell,i} + \notag \\
		&~~~~~~~~ 
		\left(1-\bm{a}_{\ell k}(i)\right)\bm{\psi}_{k,i}]
	\end{align}
	From (\ref{ATC2}), when agent $k$ is not paired with any agent at time $i$ ($\mathbf{1}_{k k} (i) = 1$), we get $\bm{w}_{k,i} = \bm{\psi}_{k,i}$. On the other hand, when agent $k$ is paired with some neighboring agent $\ell$, which means $\mathbf{1}_{k \ell} (i) = 1$, we get
	\begin{align}
		\bm{w}_{k,i} &= \alpha_k \bm{\psi}_{k,i} + (1-\alpha_k) \left[ \bm{a}_{\ell k} (i) \bm{\psi}_{\ell,i} + \left(1-\bm{a}_{\ell k}(i)\right)\bm{\psi}_{k,i} \right] 
	\end{align}
	It is then clear that $\bm{a}_{\ell k}(i)=0$ results in $\bm{w}_{k,i} = \bm{\psi}_{k,i}$, while $\bm{a}_{\ell k}(i)=1$ results in a combination of the estimates of agents $k$ and $\ell$. In other words, when $\mathbf{1}_{k \ell} (i) = 1$:
	\begin{align}
		\label{noact}
		\bm{w}_{k,i} = \begin{cases} 
			\bm{\psi}_{k,i},&\text{if}~\bm{a}_{\ell k}(i)=0 \\
			\alpha_k \bm{\psi}_{k,i} + (1-\alpha_k) \bm{\psi}_{\ell,i},&\text{if}~\bm{a}_{\ell k}(i)=1
		\end{cases}
	\end{align}
	In the sequel, we assume that agents update and combine their estimates using (\ref{ATC_A}) and (\ref{ATC2}). One important question to address is how the agents determine their actions $\{\bm{a}_{k \ell}(i)\}$.
	
	\section{Agent Interactions}
	\label{formulation}
	
	When an arbitrary agent $k$ needs to decide on whether to set its action to $\bm{a}_{k \ell}(i)=1$ (i.e., to cooperate) or $\bm{a}_{k \ell}(i)=0$ (i.e., not to cooperate), it generally cannot tell beforehand whether agent $\ell$ will reciprocate. In this section, we first show that when self-interested agents are boundedly rational and incapable of transforming the past actions of neighbors into a prediction of their future actions, then the dominant strategy for each agent will be to choose noncooperation. Consequently, the entire network becomes noncooperative. Later, in Sec.~\ref{AdRe}, we explain how to address this inefficient scenario by proposing a protocol that will encourage cooperation. 
	
	\subsection{Long-Term Discounted Cost Function}
	To begin with, let us examine the interaction between a pair of agents, such as $k$ and $\ell$, at some time instant $i$ ($\mathbf{1}_{k\ell}(i)=1$). We assume that agents $k$ and $\ell$ simultaneously select their actions $\bm{a}_{k \ell}(i)$ and $\bm{a}_{\ell k}(i)$ by using some pure strategies (i.e., agents set their action variables by using data or realizations that are available to them, such as the estimates $\{\bm{w}_{k,i-1},\bm{w}_{\ell,i-1}\}$, rather than select their actions according to some probability distributions)\footnote{In our scenario, the discrete action set $\bm{a}_{k \ell}(i) \in \{0,1\}$ will be shown to lead to threshold-based pure strategies --- see Sec.~\ref{Opt}.}. 
	The criterion for setting $\bm{a}_{k \ell}(i)$ by agent $k$ is to optimize agent $k$'s payoff, which incorporates both the estimation cost, affected by agent $\ell$'s own action $\bm{a}_{ \ell k}(i)$, and the communication cost, determined by agent $k$'s action $\bm{a}_{k \ell}(i)$. Therefore, the instantaneous cost incurred by agent $k$ is a mapping function from the action space $(\bm{a}_{k\ell}(i), \bm{a}_{\ell k}(i))$ to a real value. In order to account for selfish behavior, we need to modify the notation used in (\ref{Jcoop}) to incorporate the actions of both agents $k$ and $\ell$. In this way, we need to denote the value of the cost incurred by agent $k$ at time $i$, after $\bm{w}_{k,i-1}$ is updated to $\bm{w}_{k,i}$, more explicitly by $J_{k}(\bm{a}_{k \ell}(i),\bm{a}_{\ell k}(i))$ and it is given by:
	\begin{align}
		\label{comuti}
		J_{k}&(\bm{a}_{k \ell}(i),\bm{a}_{\ell k}(i)) \notag \\
		&= 
		\begin{cases}
			J_k^\text{est}(\bm{w}_{k,i}=\bm{\psi}_{k,i}), & \text{if}~ (0,0)\\
			J_k^\text{est}(\bm{w}_{k,i}=\alpha_k\bm{\psi}_{k,i}+(1-\alpha_k)\bm{\psi}_{\ell,i}), & \text{if}~ (0,1)\\
			J_k^\text{est}(\bm{w}_{k,i}=\bm{\psi}_{k,i})+ c_k, & \text{if}~ (1,0)\\
			J_k^\text{est}(\bm{w}_{k,i}=\alpha_k\bm{\psi}_{k,i}+(1-\alpha_k)\bm{\psi}_{\ell,i})+ c_k, & \text{if}~ (1,1)
		\end{cases} 
	\end{align}
	For example, the first line on the right-hand side of (\ref{comuti}) corresponds to the situation in which none of the agents decides to cooperate. In that case, agent $k$ can only rely on its intermediate estimate, $\bm{\psi}_{k,i}$, to improve its estimation accuracy. In comparison, the second line in (\ref{comuti}) corresponds to the situation in which agent $\ell$ is willing to share its estimate but not agent $k$. In this case, agent $k$ is able to perform the second combination step in (\ref{noact}) and enhance its estimation accuracy. In the third line in (\ref{comuti}), agent $\ell$ does not cooperate while agent $k$ does. In this case, agent $k$ incurs a communication cost, $c_k$. Similarly, for the last line in (\ref{comuti}), both agents cooperate. In this case, agent $k$ is able to perform the second step in (\ref{noact}) while incurring a cost $c_k$.
	
	We can write (\ref{comuti}) more compactly as follows:
	\begin{align}
		\label{compactJ}
		J_{k}(\bm{a}_{k \ell}(i),\bm{a}_{\ell k}(i))=J_{k}^\text{act}(\bm{a}_{ \ell k}(i))+ \bm{a}_{k \ell}(i) c_k 
	\end{align}
	where we introduced
	\begin{align}
		\label{actJ}
		J_{k}^\text{act}&(\bm{a}_{ \ell k}(i)) \notag \\
		&\triangleq \begin{cases}
			J_k^\text{est}(\bm{w}_{k,i}=\bm{\psi}_{k,i}), & \text{if~} \bm{a}_{ \ell k}(i)=0 \\
			J_k^\text{est}(\bm{w}_{k,i}=\alpha_k\bm{\psi}_{k,i}+(1-\alpha_k)\bm{\psi}_{\ell,i}), & \text{if~} \bm{a}_{ \ell k}(i)=1
		\end{cases} 
	\end{align}
	The function $J_{k}^\text{act}(\bm{a}_{\ell k}(i))$ helps make explicit the influence of the action by agent $\ell$ on the estimation accuracy that is ultimately attained by agent $k$.
	
	\begin{table*}[t!] 
		\caption{The expected long-term cost functions $J_{k,i}^1$ and $J_{\ell,i}^1$.} 
		\footnotesize
		\centering 
		
		\begin{tabu}{|[1pt]c||c|c|[1pt]}
			\thickhline
			\backslashbox[0pt]{Agent $\ell$}{Agent $k$}
			& \cellcolor[HTML]{ECF4FF}$a_{k \ell}(i)=0$ & \cellcolor[HTML]{ECF4FF}$a_{k \ell}(i)=1$ \\  
			\hline \hline 
			\cellcolor[HTML]{ECF4FF}$a_{\ell k}(i)=0$ & \multicolumn{1}{c|}{
				$\begin{aligned}
				&\mathbb{E}[J_{\ell}^\text{act}(a_{ k \ell}(i)=0)|\bm{w}_{\ell,i-1}] \\ 
				&\mathbb{E}[J_{k}^\text{act}(a_{ \ell k}(i)=0)|\bm{w}_{k,i-1}]
				\end{aligned}$}
			& $\begin{aligned}
			&\mathbb{E}[J_{\ell}^\text{act}(a_{ k \ell}(i)=1)|\bm{w}_{\ell,i-1}] \\ 
			&\mathbb{E}[J_{k}^\text{act}(a_{ \ell k}(i)=0)|\bm{w}_{k,i-1}]+c_k
			\end{aligned}$ \\ 
			\hline
			\cellcolor[HTML]{ECF4FF}$a_{\ell k} (i)=1$ & 
			$\begin{aligned}
			&\mathbb{E}[J_{\ell}^\text{act}(a_{ k \ell}(i)=0)|\bm{w}_{\ell,i-1}]+c_\ell \\ 
			&\mathbb{E}[J_{k}^\text{act}(a_{ \ell k}(i)=1)|\bm{w}_{k,i-1}]
			\end{aligned}$
			& 
			$\begin{aligned}
			&\mathbb{E}[J_{\ell}^\text{act}(a_{ k \ell}(i)=1)|\bm{w}_{\ell,i-1}]+c_\ell \\ 
			&\mathbb{E}[J_{k}^\text{act}(a_{ \ell k}(i)=1)|\bm{w}_{k,i-1}]+c_k
			\end{aligned}$ \\ \thickhline 
		\end{tabu}
		\label{table1}
	\end{table*}
	
	Now, the random-pairing process occurs repeatedly over time and, moreover, agents may leave the network. For this reason, 
	rather than rely on the instantaneous cost function in (\ref{comuti}), agent $k$ will determine its  action at time $i$ by instead minimizing an expected long-term discounted cost function of the form defined by (\ref{a_sel}) where $\delta_k \in (0,1)$ is a discount factor to model future network uncertainties and the foresightedness level of agent $k$. The expectation is taken over all randomness for $t \geq i$ and is conditioned on the estimate $\bm{w}_{k,i-1}$ when the actions $a_{k \ell}(i)$ and $a_{\ell k}(i)$ are selected. Formulation (\ref{a_sel}) is meant to assess the influence of the action selected at time $i$ by agent $k$ on its cumulative (but discounted) future costs. More specifically, whenever $\mathbf{1}_{k \ell}(i)=1$, agent $k$ selects its action $a_{k \ell}(i)$ at time $i$ to minimize the expected long-term discounted cost given $\bm{w}_{k,i-1}$:
	\setcounter{equation}{25}
	\begin{align}
		\label{min}
		\min_{a_{k \ell}(i)\in \{0,1\}} J_{k,i}^\infty [a_{k \ell}(i),a_{\ell k}(i)|\bm{w}_{k,i-1}]
	\end{align}
	Based on the payoff function in (\ref{a_sel}), we can formally regard the interaction between agents as consisting of stage games with recurrent random pairing. The stage information-sharing game for $\mathbf{1}_{k\ell}(i)=1$ is a tuple $(\mathds{N} ,\mathds{A},\mathds{J})$, where $\mathds{N}\triangleq \{k,\ell\}$ is the set of players, and $\mathds{A} \triangleq \mathds{A}_k \times \mathds{A}_\ell$ is the Cartesian product of binary sets $\mathds{A}_k=\mathds{A}_\ell \triangleq \{1,0\}$ representing available actions for agents $k$ and $\ell$, respectively. The action profile is $a(i) \triangleq (a_{k \ell}(i),a_{\ell k}(i)) \in \mathds{A}$. Moreover, $\mathds{J} = \{J_{k,i}^\infty,J_{\ell,i}^\infty\}$ is the set of real-valued long-term costs defined over $\mathds{A} \rightarrow \mathbb{R}$ for agents $k$ and $\ell$, respectively. 
	We remark that since $J_{k,i}^\infty$ depends on $\bm{w}_{k,i-1}$, its value generally varies from stage to stage. As a result, each agent $k$ faces a dynamic game structure with repeated interactions in contrast to conventional repeated games as in~\cite{Shoham08,Xu14} where the game structure is fixed over time. Time variation is an essential feature that arises when we examine selfish behavior over adaptive networks.
	
	Therefore, solving problem (\ref{min}) involves the forecast of future game structures and future actions chosen by the opponent. These two factors are actually coupled and influence each other; this fact makes prediction under such conditions rather challenging. To continue with the analysis, we adopt a common assumption from the literature that agents have computational constraints. In particular, we assume the agents have bounded rationality~\cite{Gigerenzer02,Simon55,Park09}. In our context, this means that the agents have limited capability to forecast future game structures and are therefore obliged to assume that future parameters remain unchanged at current values. We will show how this assumption enables each agent $k$ to evaluate $J_{k,i}^\infty$ in later discussions.
	\begin{asump} [Bounded rationality] 
		\label{ass:rationalAgents}
		Every agent $k$ solves the optimization problem (\ref{min}) under the assumptions:
		\begin{align}
			\label{assump1}
			\bm{w}_{k,t}=\bm{w}_{k,i-1},~~\mathbf{1}_{k \ell}(t)=\mathbf{1}_{k \ell}(i),~~\text{for}~t\geq i
		\end{align}
		\hfill $\blacksquare$
		\vspace{2mm}
	\end{asump}
	We note that the above assumption is only made 
	by the agent at time $i$ while solving problem (\ref{min}); the actual estimates $\bm{w}_{k,t}$ and pairing choices $\mathbf{1}_{k\ell}(t)$ will continue to evolve over time. We further assume that the bounded rationality assumption is common knowledge to all agents in the network\footnote{Common knowledge of $p$ means that each agent knows $p$, each agent knows that all other agents know $p$, each agent knows that all other agents know that all the agents know $p$, and so on~\cite{Osborne94}.}.
	
	\subsection{Pareto Inefficiency}
	\label{dominant}
	In this section, we show that if no further measures are taken, then Pareto inefficiency may occur. Thus, assume that the agents are unable to store the history of their actions and the actions of their neighbors. Each agent $k$ only has access to its immediate estimate $\bm{w}_{k,i-1}$, which can be interpreted as a state variable at time $i-1$ for agent $k$. 
	In this case, each agent $k$ will need to solve (\ref{min}) under Assumption \ref{ass:rationalAgents}.
	It then follows that agent $k$ will predict the same action for future time instants:
	\begin{align}
		\label{prison1}
		\bm{a}_{k \ell}(t) = a_{k \ell}(i),~~\text{for}~t> i
	\end{align}
	Furthermore, since the bounded rationality condition is common knowledge, agent $k$ knows that the same future actions are used by agent $\ell$, i.e.,
	\begin{align}
		\label{prison2}
		\bm{a}_{\ell k}(t) = a_{\ell k}(i),~~\text{for}~t> i
	\end{align}
	Using (\ref{prison1}) and (\ref{prison2}), agent $k$ obtains
	\begin{align}
		\label{oneJ}
		J_{k,i}^\infty &\left[a_{k \ell}(i),a_{\ell k}(i)|\bm{w}_{k,i-1}\right] \notag \\
		&=\sum_{t=i}^{\infty} \delta_k^{t-i} \mathbb{E} 
		\Big[
		J_{k}(a_{k \ell}(i),a_{\ell k}(i)) \Big{|}\bm{w}_{k,i-1}\Big] \notag \\
		&=\frac{1}{1-\delta_k} \cdot \mathbb{E} \Big[J_{k}(a_{k \ell}(i),a_{ \ell k}(i))\Big{|}\bm{w}_{k,i-1} \Big] \notag \\
		&=\frac{1}{1-\delta_k} \left(\mathbb{E} [
		J_{k}^\text{act}(a_{ \ell k}(i))|\bm{w}_{k,i-1}]+a_{k \ell}(i) c_k \right) 
	\end{align}
	Therefore, the optimization problem (\ref{min}) reduces to the following minimization problem:
	\begin{align}
		\label{min2}
		\min_{a_{k \ell}(i)\in \{0,1\}} J_{k,i}^1 \left(a_{k \ell}(i),a_{\ell k}(i)\right)
	\end{align}
	where 
	\begin{align}
		\label{oneshot}
		\!\!\!J_{k,i}^1 \left(a_{k \ell}(i),a_{\ell k}(i)\right) \triangleq \mathbb{E} [
		J_{k}^\text{act}(a_{ \ell k}(i))|\bm{w}_{k,i-1}]+a_{k \ell}(i) c_k
	\end{align}
	is the expected cost of agent $k$ given $\bm{w}_{k,i-1}$ --- compare with (\ref{compactJ}).
	Table~\ref{table1} summarizes the values of $J_{k,i}^1$ and $J_{\ell,i}^1$ for both agents under their respective actions. From the entries in the table, we conclude that choosing action $a_{k \ell}(i)=0$ is the dominant strategy for agent $k$ regardless of the action chosen by agent $\ell$ because its cost will be the smallest it can be in that situation. Likewise, the dominant strategy for agent $\ell$ is $a_{\ell k}(i)=0$ regardless of the action chosen by agent $k$. Therefore, the action profile $(a_{k \ell}(i),a_{\ell k}(i))=(0,0)$ is the unique outcome as a Nash and dominant strategy equilibrium for every stage game. 
	
	\begin{figure*}[!t]
		\centerline{\subfloat[$\bm{\gamma}_{k,i}>1$ and $\bm{\gamma}_{\ell,i}>1$]{\includegraphics[width=2.8in]{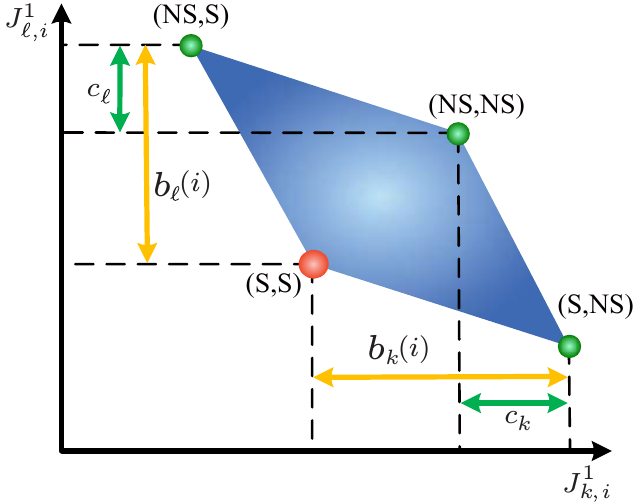}
				\label{fig:a}}
			\hfil
			\subfloat[$\bm{\gamma}_{k,i}<1$ and $\bm{\gamma}_{\ell,i}<1$]{\includegraphics[width=2.8in]{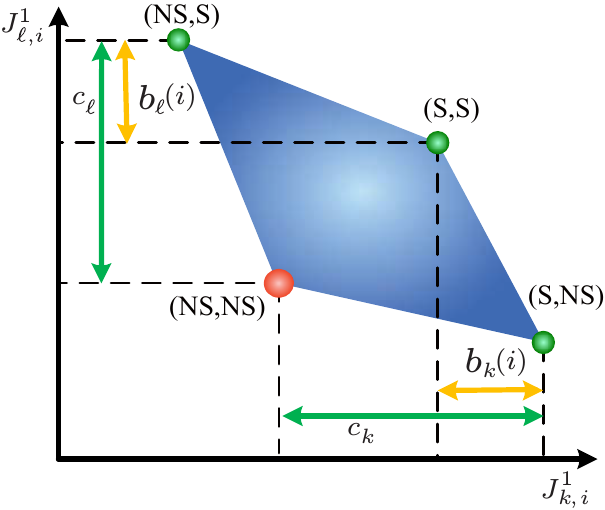}
				\label{fig:b}}}
		\caption{Illustration of the behavior of the payoffs in terms of the size of the benefit-cost ratios (``S" and ``NS" refer to the actions ``share" and ``do not share", respectively).}
		\label{fig:CC}
	\end{figure*}
	
	However, this resulting action profile will be Pareto inefficient for both agents if it can be verified that
	the alternative action profile $(1,1)$, where both agents cooperate, can lead to improved payoff values for both agents in comparison to the strategy $(0,0)$. 
	To characterize when this is possible, let us denote the expected payoff for agent $k$ when agent $\ell$ selects $a_{\ell k}(i)=0$ by 
	\begin{align}
		\label{sk}
		\bm{s}_{k,i}^0(a_{k \ell}(i)) &\triangleq  \mathbb{E}[J_{k}^\text{act}(a_{\ell k}(i)=0)|\bm{w}_{k,i-1}]+a_{k \ell}(i) c_k 
	\end{align}
	Likewise, when $a_{\ell k}(i)=1$, we denote the expected payoff for
	agent $k$ by
	\begin{align}
		\label{sl}
		\bm{s}_{k,i}^1(a_{k \ell}(i)) &\triangleq \mathbb{E}[J_{k}^\text{act}(a_{\ell k}(i)=1)|\bm{w}_{k,i-1}]+a_{k \ell}(i) c_k 
	\end{align}
	The benefit for agent $k$ from agent $\ell$'s sharing action, defined as the improvement from $\bm{s}_{k,i}^0(a_{k \ell}(i))$ to $\bm{s}_{k,i}^1(a_{k \ell}(i))$, is seen to be independent of $a_{k \ell}(i)$:
	\begin{align}
		\label{benefit}
		&\bm{b}_k(i) \notag \\
		&~\triangleq \bm{s}_{k,i}^0(a_{k \ell}(i))-\bm{s}_{k,i}^1(a_{k \ell}(i)) \notag \\
		&~=\mathbb{E}[J_{k}^\text{act}(a_{\ell k}(i)=0)|\bm{w}_{k,i-1}]- \mathbb{E}[J_{k}^\text{act}(a_{\ell k}(i)=1)|\bm{w}_{k,i-1}] \notag \\
		&~=\mathbb{E}\left[J_k^\text{est}\left(\bm{w}_{k,i}=\bm{\psi}_{k,i}\right)|\bm{w}_{k,i-1}\right] \notag \\
		&~~~~-\mathbb{E}\left[J_k^\text{est}\left(\bm{w}_{k,i}=\alpha_k\bm{\psi}_{k,i}+(1-\alpha_k)\bm{\psi}_{\ell,i}\right)|\bm{w}_{k,i-1}\right]
	\end{align}
	Now, note from definition (\ref{apriori}) that
	\begin{align}
		\label{condMSE}
		&\mathbb{E} \left[J_k^\text{est}(\bm{w}_{k,i}) | \bm{w}_{k,i-1} \right] \notag \\
		&= \mathbb{E} \left[| \bm{d}_k(i+1) - \bm{u}_{k,i+1} \bm{w}_{k,i} |^2 \Big| \bm{w}_{k,i-1} \right]
	\end{align}
	so that
	\begin{align}
		\label{BPred}
		\mathbb{E}[&J_{k}^\text{act}(a_{\ell k}(i)=0)|\bm{w}_{k,i-1}] \notag \\ &~~=\mathbb{E}\left[J_k^\text{est}\left(\bm{w}_{k,i}=\bm{\psi}_{k,i}\right)|\bm{w}_{k,i-1}\right] \notag \\
		&~~=\mathbb{E}\left[|\bm{d}_k(i+1)-\bm{u}_{k,i+1}\bm{\psi}_{k,i}|^2
		\Big|\bm{w}_{k,i-1}\right] \notag \\
		&~~=\mathbb{E}\left[|\bm{u}_{k,i+1} \widetilde{\bm{\psi}}_{k,i}+\bm{v}_k(i+1)|^2
		\Big|\bm{w}_{k,i-1}\right] \notag \\
		&~~=\mathbb{E}\left[\|\widetilde{\bm{\psi}}_{k,i}\|^2_{R_{u,k}}
		\Big|\bm{w}_{k,i-1}\right] +\sigma^2_{v,k}
	\end{align}
	where $\widetilde{\bm{\psi}}_{k,i} \triangleq w^o-\bm{\psi}_{k,i}$ and, similarly,
	\begin{align}
		\label{BPred1}
		&\mathbb{E}[J_{k}^\text{act}(a_{\ell k}(i)=1)|\bm{w}_{k,i-1}] \notag \\
		&~~=\mathbb{E}\left[J_k^\text{est}\left(\bm{w}_{k,i}=\alpha_k\bm{\psi}_{k,i}+(1-\alpha_k)\bm{\psi}_{\ell,i}\right)|\bm{w}_{k,i-1}\right] \notag \\
		&~~=\mathbb{E}\left[\|\alpha_k \widetilde{\bm{\psi}}_{k,i}+(1-\alpha_k)\widetilde{\bm{\psi}}_{\ell,i}\|^2_{R_{u,k}}
		\Big|\bm{w}_{k,i-1}\right] +\sigma^2_{v,k}
	\end{align}
	Then, the benefit $\bm{b}_k(i)$ becomes
	\begin{align}
		\label{benefit2}
		\bm{b}_k(i) &= \mathbb{E}\left[\|\widetilde{\bm{\psi}}_{k,i}\|^2_{R_{u,k}}\Big|\bm{w}_{k,i-1}\right] \notag \\
		&~~~-\mathbb{E}\left[\|\alpha_k \widetilde{\bm{\psi}}_{k,i}+(1-\alpha_k)\widetilde{\bm{\psi}}_{\ell,i}\|^2_{R_{u,k}}
		\Big|\bm{w}_{k,i-1}\right]
	\end{align}
	Note that $\bm{b}_k(i)$ is determined by the variable $\bm{w}_{k,i-1}$ and does not depend on the actions $a_{\ell k}(i)$ and $a_{k \ell}(i)$. We will explain how agents assess the information $\bm{b}_k(i)$ to choose actions further ahead in Sec. IV-C.
	Now, let us define the benefit-cost ratio as the ratio of the estimation benefit to the communication cost:
	\begin{align}
		\bm{\gamma}_k(i) \triangleq \frac{\bm{b}_k(i)}{c_k}
	\end{align}
	Then, the action profile $(1,1)$ in the game defined in Table~\ref{table1} is Pareto superior to the action profile $(0,0)$ when both of the following two conditions hold
	\begin{align}
		\label{gamma}
		\bm{\gamma}_k(i) >1~~\text{and}~~\bm{\gamma}_\ell(i) >1 ~~\Leftrightarrow~~
		\begin{cases}
			c_k < \bm{b}_k(i)\\
			c_\ell < \bm{b}_\ell(i)
		\end{cases} 
	\end{align}
	On the other hand, the action profile $(0,0)$ is Pareto superior to the action profile $(1,1)$ if, and only if,
	\begin{align}
		\label{gamma2}
		\bm{\gamma}_k(i) <1~~\text{and}~~\bm{\gamma}_\ell(i) <1
	\end{align}
	In Fig.~\ref{fig:CC}(a), we illustrate how the values of the payoffs compare to each other when (\ref{gamma}) holds for the four possibilities of action profiles. It is seen from this figure that when $\bm{\gamma}_k(i)>1$ and $\bm{\gamma}_\ell(i)>1$, the action profile (S,S), i.e., $(1,1)$ in (\ref{oneshot}), is Pareto optimal and that the dominant strategy (NS,NS), i.e., $(0,0)$ in (\ref{oneshot}), is inefficient and leads to worse performance (which is a manifestation of the famous prisoner's dilemma problem~\cite{Rapaport65}). On the other hand, if $\bm{\gamma}_k(i)<1$ and $\bm{\gamma}_\ell(i)<1$, then we are led to Fig.~\ref{fig:CC}(b), where the action profile (NS,NS) becomes Pareto optimal and superior to (S,S). We remark that (NS,S) and (S,NS) are also Pareto optimal in both cases but not preferred in this work because they are only beneficial for one single agent. 
	
	\section{Adaptive Reputation Protocol Design}
	\label{AdRe}
	As shown above, when both $\bm{\gamma}_k(i)>1$ and $\bm{\gamma}_\ell(i)>1$, the Pareto optimal strategies for agents $k$ and $\ell$ correspond to cooperation; when both $\bm{\gamma}_k(i)<1$ and $\bm{\gamma}_\ell(i)<1$, the Pareto optimal strategies for agents $k$ and $\ell$ reduce to non-cooperation. Since agents are self-interested and boundedly rational, we showed earlier that if left without incentives, their dominant strategy is to avoid sharing information because they cannot tell beforehand if their paired neighbor will reciprocate. This Pareto inefficiency therefore arises from the fact that agents are not using historical data to predict other agents' actions. We now propose a reputation protocol to summarize the opponent's past actions into a reputation score. The score will help agents to form a belief of their opponent's subsequent actions. Based on this belief, we will be able to provide agents with a measure that entices them to cooperate. We will show, for example, that the best response rule for agents will be to cooperate whenever $\bm{\gamma}_k(i)$ is large and not to cooperate whenever $\bm{\gamma}_k(i)$ is small, in conformity with the Pareto-efficient design.

	\setcounter{equation}{48}
	\begin{figure*}[!b]
		\hrulefill
		\normalsize
		\begin{align}
			\label{Pk}
			J_{k,i}^{\infty' }[a_{k\ell}(i)|\bm{w}_{k,i-1}]
			&=\!\sum_{t=i}^{\infty} \delta_k^{t-i} \mathbb{E} \!
			\left[
			J_{k}^\text{act}(\bm{a}_{\ell k}(t))+ \bm{a}_{k \ell}(t) c_k \Big{|} 
			\bm{w}_{k,i-1}, \bm{a}_{k \ell}(i)=a_{k \ell}(i), \mathbb{K}_i \right] \notag \\
			&=\!\mathbb{E} 
			\!\left[
			J_{k}^\text{act}(\bm{a}_{\ell k}(i))\Big{|} 
			\bm{w}_{k,i-1}, \mathbb{K}_i \right]\!+\!a_{k \ell}(i) c_k +\!\!
			\sum_{t=i+1}^{\infty} \delta_k^{t-i} \mathbb{E} \!
			\left[
			J_{k}^\text{act}(\bm{a}_{\ell k}(t))+ \bm{a}_{k \ell}(t) c_k \Big{|} 
			\bm{w}_{k,i-1}, \bm{a}_{k \ell}(i)=a_{k \ell}(i), \mathbb{K}_i \right]
		\end{align}
	\end{figure*}
	\setcounter{equation}{42}
	
	\subsection{Reputation Protocol}
	
	Reputation scores have been used before in the literature as a mechanism to encourage cooperation~\cite{Zhang12,Jurca04,Fan05}. Agents that cooperate are rewarded with higher scores; agents that do not cooperate are penalized with lower scores.
	For example, eBay uses a cumulative score mechanism, which simply sums the seller’s feedback scores from all previous periods to provide buyers and sellers with trust evaluation~\cite{Houser06}. Likewise, Amazon.com implements a reputation system by using an average score mechanism that averages the feedback scores from the previous periods~\cite{Audun07}. However, as already explained in~\cite{Fan05}, cheating can occur over time in both cumulative and average score mechanisms because past scores carry a large weight in determining the current reputation. To overcome this problem, and in a manner similar to exponential weighting in adaptive filter designs~\cite{Sayed08}, an exponentially-weighted moving average mechanism that gives higher weights to more recent actions is discussed in~\cite{Fan05}. We follow a similar weighting formulation, with the main difference being that the reputation scores now need to be adapted in response to the evolution  of the estimation task over the network. The construction can be described as follows. 
	
	When $\mathbf{1}_{k \ell}(i)=1$, meaning that agent $k$ is paired with agent $\ell$, the reputation score $\bm{\theta}_{\ell k}(i) \in [0,1]$ that is maintained by agent $k$ for its neighbor $\ell$ is updated as:
	\begin{align}
		\label{repuo}
		\bm{\theta}_{\ell k}(i+1) = r_k \bm{\theta}_{\ell k}(i) +(1-r_k)\bm{a}_{\ell k}(i)
	\end{align}
	where $r_k\in (0,1)$ is a smoothing factor for agent $k$ to control the dynamics of the reputation updates. On the other hand, if $\mathbf{1}_{k \ell}(i)=0$, the reputation score $\bm{\theta}_{\ell k}(i+1)$ remains as $\bm{\theta}_{\ell k}(i)$. We can compactly describe the reputation rule as
	\begin{align}
		\label{repu}
		\bm{\theta}_{\ell k}(i+1) &= \mathbf{1}_{k \ell}(i) \left[ r_k \bm{\theta}_{\ell k}(i) +(1-r_k)\bm{a}_{\ell k}(i)\right] \notag \\
		&~~~+(1-\mathbf{1}_{k \ell}(i)) \bm{\theta}_{\ell k}(i)
	\end{align}
	Directly applying the above reputation formulation, however, 
	can cause a loss in adaptation ability over the network. For example, the network would become permanently non-cooperative when agent $\ell$ chooses $\bm{a}_{\ell k}(i)=0$ for long consecutive iterations. That is because, in that case, the reputation score $\bm{\theta}_{\ell k}(i)$ will decay exponentially to zero, which keeps agent $k$ from choosing $\bm{a}_{k \ell}(i)=1$ in the future. 
	In order to avoid this situation, we set a lowest value for the reputation score to a small positive threshold $0<\varepsilon \ll 1$, i.e.,
	\begin{align}
		\label{repu2}
		\bm{\theta}_{\ell k}(i+1) &= \mathbf{1}_{k \ell}(i) \cdot \max\{ r_k \bm{\theta}_{\ell k}(i) +(1-r_k)\bm{a}_{\ell k}(i),\varepsilon \}\notag \\
		&~~~+ (1-\mathbf{1}_{k \ell}(i)) \bm{\theta}_{\ell k}(i)
	\end{align}
	and thus $\bm{\theta}_{\ell k}(i)\in [\varepsilon,1]$.
	
	The reputation scores can now be utilized to evaluate the belief 
	by agent $k$ of subsequent actions by agent $\ell$. To explain how this can be done, we argue that agent $k$ would expect the probability of $\bm{a}_{\ell k}(t)=1$, i.e., the probability that agent $\ell$ is willing to cooperate, to be an increasing function of both $\bm{\theta}_{\ell k}(t)$ and $\bm{\theta}_{k \ell}(t)$ for $t \geq i$. Specifically, if we denote this belief probability by $B(\bm{a}_{\ell k}(t)=1)$, then it is expected to satisfy:
	\begin{align}
		\label{Bformulate}
		\frac{\partial B(\bm{a}_{\ell k}(t)=1 )}{\partial \bm{\theta}_{\ell k}(t)} \geq 0,~~ \frac{\partial B(\bm{a}_{\ell k}(t)=1 )}{\partial \bm{\theta}_{k \ell}(t)} \geq 0
	\end{align}
	The first property is motivated by the fact that according to the history of actions, a higher value for $\bm{\theta}_{\ell k}(t)$ indicates that agent $\ell$ has higher willingness to share estimates. The second property is motivated by the fact that lower values for $\bm{\theta}_{k \ell}(t)$ mean that agent $k$ has rarely shared estimates with agent $\ell$ in the recent past. Therefore, it can be expected that agent $\ell$ will have lower willingness to share information for lower values of $\bm{\theta}_{k\ell}(t)$. Based on this argument, we suggest a first-order construction for measuring belief with respect to both $\bm{\theta}_{\ell k}(t)$ and $\bm{\theta}_{k \ell}(t)$ as follows (other constructions are of course possible; our intent is to keep the complexity of the solution low while meeting the desired objectives):
	\begin{align}
		\label{probb}
		B(\bm{a}_{\ell k}(t)=1 ) =  \bm{\theta}_{k \ell}(t) \cdot \bm{\theta}_{\ell k}(t), ~~ t \geq i
	\end{align}
	which satisfies both properties in (\ref{Bformulate}) and where $B(\bm{a}_{\ell k}(t)=1 ) \in [\varepsilon^2,1]$. 
	Therefore, the reputation protocol implements (\ref{repu2}) and (\ref{probb}) repeatedly. Each agent $k$ will then employ the reference knowledge $\mathbb{K}_i\triangleq \{\bm{\theta}_{k \ell}(i),\bm{\theta}_{\ell k}(i), B(\bm{a}_{\ell k}(i)=1 )\}$ to select its action $a_{k \ell}(i)$ as described next. 
	
	\subsection{Best Response Rule}
	\label{Opt}
	The belief measure (\ref{probb}) provides agent $k$ with additional information about agent  $\ell$'s actions. That is, with (\ref{probb}), agent $k$ can treat $\bm{a}_{\ell k}(t)$ as a random variable with distribution $B(\bm{a}_{\ell k}(t)=1)$ for $t\geq i$. 
	Then, the best response of agent $k$ is obtained by solving the following optimization problem:
	\begin{align}
		\label{min3}
		\min_{a_{k \ell}(i)\in \{0,1\}} J_{k,i}^{\infty' }[a_{k\ell}(i)|\bm{w}_{k,i-1}] 
	\end{align}
	where $J_{k,i}^{\infty' }[a_{k\ell}(i)|\bm{w}_{k,i-1}]$ is defined by (\ref{Pk}) and involves an additional expectation over the distribution of $\bm{a}_{\ell k}(t)$ --- compare with (\ref{a_sel}). Similarly to Assumption \ref{ass:rationalAgents}, we assume the bounded rationality of the agents extends to the reputation scores $\bm{\theta}_{\ell k}(t)$ for $t \geq i$.
	\vspace{2mm}
	\begin{asump} [Extended bounded rationality] 
		\label{ass:rationalAgentsEx}
		We extend the assumption of bounded rationality from (\ref{assump1}) to also include:
		\begin{align}
			\setcounter{equation}{49}
			\label{assump2}
			\bm{\theta}_{\ell k}(t)=\bm{\theta}_{\ell k}(i),~~\text{for}~t\geq i
		\end{align}~
		\hfill $\blacksquare$
		\vspace{2mm}
	\end{asump}
	Now, using pure strategies, the best response of agent $k$ is to select the action $a_{k \ell}(i)$ such that
	\begin{align}
		\label{best}
		a_{k \ell}(i)=\begin{cases}
			1, & \text{if}~ J_{k,i}^{\infty' }[a_{k\ell}(i)=1|\bm{w}_{k,i-1}]  \\
			&~~~~~~~~~<J_{k,i}^{\infty' }[a_{k\ell}(i)=0|\bm{w}_{k,i-1}]\\
			0, & \text{otherwise} 
		\end{cases}
	\end{align}
	The following lemma shows how the best response rule depends on the benefit-cost ratio $\bm{\gamma}_k(i)$ and the communication cost $c_k$:
	\vspace{2mm}
	\begin{lemma}
		With Assumptions 1 and 2, the best response rule $f_k(\cdot)$ becomes
		\begin{align}
			\label{stra}
			a_{k \ell}(i)=
			\begin{cases}
				1, & \text{if}~ \bm{\gamma}_k(i)\triangleq \frac{\bm{b}_k(i)}{c_k}> \frac{\chi_k}{\bm{\theta}_{\ell k}(i)}\\
				0, & \text{otherwise} 
			\end{cases}
		\end{align}
		where 
		\begin{align}
			\label{chi}
			\chi_k& \triangleq \frac{1- \delta_k r_k }{\delta_k(1-r_k)}
		\end{align} 
	\end{lemma}
	\begin{proof} 
		See Appendix A.
	\end{proof}
	\vspace{2mm}
	We note that the resulting rule aligns the agents to achieve the Pareto optimal strategy: to share information when $\bm{\gamma}_k(i)$ is sufficiently large and not to share information when $\bm{\gamma}_k(i)$ is small.	 
			
	\subsection{Benefit Prediction}
	To compute the benefit-cost ratio $\bm{\gamma}_k(i)=\bm{b}_k(i)/c_k$, the agent still needs to know $\bm{b}_k(i)$ defined by (\ref{benefit}), which depends on the quantities $\mathbb{E}[J_{k}^\text{act}(a_{\ell k}(i)=0)|\bm{w}_{k,i-1}]$ and $\mathbb{E}[J_{k}^\text{act}(a_{\ell k}(i)=1)|\bm{w}_{k,i-1}]$.
	It is common in the literature, as in~\cite{Namvar13,Nayyar13}, to assume that agents have complete information about the payoff functions like the ones shown in Table~\ref{table1}. However, in the context of adaptive networks where agents have only access to data realizations and not to their statistical distributions, the payoffs are unknown and need to be estimated or predicted. For example, in our case, the convex combination $\alpha_k\bm{\psi}_{k,i}+(1-\alpha_k)\bm{\psi}_{\ell,i}$ is unknown for agent $k$ before agent $\ell$ shares $\bm{\psi}_{\ell,i}$ with it. We now describe one way by which agent $k$ can predict $\bm{b}_k(i)$; other ways are possible depending on how much information is available to the agent. 
	Let us assume a special type of agents, which are called risk-taking~\cite{Wyatt90}: agent $k$ chooses $a_{k\ell}(i)=1$ as long as the largest achievable benefit, denoted by $\bar{\bm{b}}_k(i)$, exceeds the threshold: 
	\begin{align}
		\label{stra2}
		a_{k \ell}(i)=
		\begin{cases}
			1, & \text{if}~ \frac{\bar{\bm{b}}_k(i)}{c_k}> \frac{\chi_k}{\bm{\theta}_{\ell k}(i)} \\
			0, & \text{otherwise} 
		\end{cases}
	\end{align}
	Using (\ref{benefit2}), the largest achievable benefit $\bar{\bm{b}}_k(i)$ can be found by solving the following optimization problem:
	\begin{align}
		\label{bbar}
		\bar{\bm{b}}_k(i) &\triangleq \max_{\widetilde{\bm{\psi}}_{\ell,i}} \bm{b}_k(i) \notag \\
		&= \max_{\widetilde{\bm{\psi}}_{\ell,i}} 
		\Big\{\mathbb{E}\left[\|\widetilde{\bm{\psi}}_{k,i}\|^2_{R_{u,k}}\Big|\bm{w}_{k,i-1}\right] \notag \\
		&~~~~~~~~~~~-\mathbb{E}\left[\|\alpha_k \widetilde{\bm{\psi}}_{k,i}+(1-\alpha_k)\widetilde{\bm{\psi}}_{\ell,i}\|^2_{R_{u,k}}
		\Big|\bm{w}_{k,i-1}\right] \Big\} \notag \\
		&=\mathbb{E}\left[\|\widetilde{\bm{\psi}}_{k,i}\|^2_{R_{u,k}} \Big|\bm{w}_{k,i-1}\right]
	\end{align}
	since the maximum occurs when
	\begin{align}
		\widetilde{\bm{\psi}}_{\ell,i}=-\frac{\alpha_k}{1-\alpha_k}\widetilde{\bm{\psi}}_{k,i}
	\end{align}
	Let us express the adaptation step (\ref{ATC_A}) in terms of the estimation error as
	\begin{align}
		\label{adapt_error}
		\widetilde{\bm{\psi}}_{k,i} &=(I-\mu \bm{u}_{k,i}^* \bm{u}_{k,i}) \widetilde{\bm{w}}_{k,i-1}-\mu \bm{u}_{k,i}^* \bm{v}_k(i)
	\end{align}
	To continue, we assume that the step-size $\mu$ is sufficiently small. Then, 
	\begin{align}
		\mathbb{E}&\left[\|\widetilde{\bm{\psi}}_{k,i}\|^2_{R_{u,k}}
		\Big|\bm{w}_{k,i-1}\right] \notag \\
		&=\mathbb{E}\Big[\|(I-\mu \bm{u}_{k,i}^* \bm{u}_{k,i} ) \widetilde{\bm{w}}_{k,i-1}-\mu \bm{u}_{k,i}^* \bm{v}_k(i)\|^2_{R_{u,k}} \Big|\bm{w}_{k,i-1}\Big]  \notag \\
		&=\mathbb{E}\Big[\widetilde{\bm{w}}_{k,i-1}^*(I-\mu \bm{u}_{k,i}^* \bm{u}_{k,i}) R_{u,k} (I-\mu \bm{u}_{k,i}^* \bm{u}_{k,i})  \notag \\
		&~~~~~~~ \times \widetilde{\bm{w}}_{k,i-1} \Big|\bm{w}_{k,i-1} \Big] +
		\mu^2 \text{Tr}(R_{u,k}^2) \sigma^2_{v,k}
		\notag \\
		&= \widetilde{\bm{w}}_{k,i-1}^* \Omega_k \widetilde{\bm{w}}_{k,i-1}+O(\mu^2)
	\end{align}
	where we are collecting terms that are second-order in the step-size into the factor $O(\mu^2)$,\footnote{This approximation simplifies the algorithm construction. However, when we study the network performance later in (\ref{gpsi}) we shall keep the second-order terms.} and where we introduced $\Omega_k \triangleq R_{u,k} (I-2\mu R_{u,k})$. We note that for sufficiently small step-sizes: 
	\begin{align}
		\label{Omega}
		\Omega_k' \triangleq R_{u,k} (I-\mu R_{u,k})^2 \approx R_{u,k} (I-2\mu R_{u,k}) = \Omega_k
	\end{align}
	Therefore, each agent $k$ can approximate $\bar{\bm{b}}_k(i)$ as
	\begin{align}
		\label{bhat}
		\bar{\bm{b}}_k(i) &= \widetilde{\bm{w}}_{k,i-1}^* \Omega_k'  \widetilde{\bm{w}}_{k,i-1} \notag \\
		&=\widetilde{\bm{w}}_{k,i-1}^* R_{u,k} (I-\mu R_{u,k})^2  \widetilde{\bm{w}}_{k,i-1} \notag \\
		&=\|(I-\mu R_{u,k})  \widetilde{\bm{w}}_{k,i-1}\|^2_{R_{u,k}} 
	\end{align}
	
	\subsection{Real-Time Implementation}
	\label{realtime}
	Expression (\ref{bhat}) is still not useful for agents because it requires knowledge of both $R_{u,k}$ and $\widetilde{\bm{w}}_{k,i-1}$. With regards to $R_{u,k}$, we can use the instantaneous approximation $R_{u,k} \approx \bm{u}_{k,i}^* \bm{u}_{k,i}$ to get
	\begin{align}
		\label{bapprox}
		\bar{\bm{b}}_k(i) &\approx  \widetilde{\bm{w}}_{k,i-1}^* \bm{u}_{k,i}^* \bm{u}_{k,i} (I-\mu\bm{u}_{k,i}^* \bm{u}_{k,i})^2 \widetilde{\bm{w}}_{k,i-1} \notag \\
		&=(1-\mu\|\bm{u}_{k,i}\|^2)^2 \widetilde{\bm{w}}_{k,i-1}^* \bm{u}_{k,i}^* \bm{u}_{k,i} \widetilde{\bm{w}}_{k,i-1}
	\end{align}
	With regards to $\widetilde{\bm{w}}_{k,i-1}$, we assume that agents use a moving-average filter as in~\cite{Rort10} to approximate $w^o$ iteratively as follows:
	\begin{align}
		\widehat{\bm{w}}_{k,i}^o &= (1-\nu) \widehat{\bm{w}}_{k,i-1}^o+\nu \bm{\psi}_{k,i} \\
		\label{insapp}
		\widetilde{\bm{w}}_{k,i-1} &\approx \widehat{\bm{w}}_{k,i}^o-\bm{w}_{k,i-1}
	\end{align}
	where $\nu\in (0,1)$ is a positive forgetting factor close to 0 to give higher weights on recent results. We summarize the operation of the resulting algorithm in the following listing.
	
	\begin{algorithm}
		\footnotesize
		\caption{\small Diffusion Strategy with an Adaptive Reputation Scheme.}
		\label{alg:summary}
		\begin{algorithmic}
			\renewcommand{\algorithmiccomment}[1]{{\bf{\% #1}}}
			\Statex
			\State Let $\{\bm{w}_{k,-1}=0\}$ and $\{\theta_{k \ell}(-1)=1\}$ for all $k$ and $\ell$. Define $\chi_k \triangleq \frac{1- \delta_k r_k }{\delta_k(1-r_k)}$.
			\Loop
			\State Generate $\{\mathbf{1}_{k \ell}(i)\}$ for all $k$ and $\ell$.
			\State \Comment{Stage 1 (Adaptation)}: For each $k$:
			\vspace{-2mm}
			\begin{align}
				\bm{\psi}_{k,i} &= \bm{w}_{k,i-1} + \mu \bm{u}^*_{k,i}[\bm{d}_k(i)-\bm{u}_{k,i}\bm{w}_{k,i-1}] \notag \\
				\bar{\bm{b}}_k(i) &\approx  (1-\mu\|\bm{u}_{k,i}\|^2)^2 \widetilde{\bm{w}}_{k,i-1}^* \bm{u}_{k,i}^* \bm{u}_{k,i} \widetilde{\bm{w}}_{k,i-1} \notag \\
				\widehat{\bm{w}}_{k,i}^o &= (1-\nu) \widehat{\bm{w}}_{k,i-1}^o+\nu \bm{\psi}_{k,i} \notag \\
				\widetilde{\bm{w}}_{k,i-1} &\approx \widehat{\bm{w}}_{k,i}^o-\bm{w}_{k,i-1} \notag
			\end{align} 
			\vspace{-4mm}
			\State \Comment{Stage 2 (Action Selection)}: For all $k$ and $\ell$,	
			\If{$\mathbf{1}_{k \ell}(i)=1$} 
			\State 	
			$\bm{a}_{k \ell}(i)=
			\begin{cases}
			1, & \text{if}~ \frac{\bar{\bm{b}}_k(i)}{c_k}> \frac{\chi_k}{\bm{\theta}_{\ell k}(i)}\\
			0, & \text{otherwise} 
			\end{cases}$
			\vspace{1mm}  
			\Else{ $\bm{a}_{k \ell}(i)=0$.}
			\EndIf
			\vspace{1mm}
			\State \Comment{Stage 3 (Reputation Update)}: For all $k$ and $\ell$,
			\vspace{-2mm}
			\begin{align}
				\bm{\theta}_{\ell k}(i+1) = &\mathbf{1}_{k \ell}(i) \cdot \max\{ r_k \bm{\theta}_{\ell k}(i) +(1-r_k)\bm{a}_{\ell k}(i),\varepsilon \}\notag \\
				&+ (1-\mathbf{1}_{k \ell}(i)) \bm{\theta}_{\ell k}(i) \notag
			\end{align}
			\vspace{-4mm}    
			\State \Comment{Stage 4 (Combination)}: For all $k$,
			\vspace{-2mm}
			\begin{align}
				\bm{w}_{k,i} = \alpha_k \bm{\psi}_{k,i} \!+\! (1-\alpha_k)\!\!\sum_{\ell \in \mathcal{N}_k}\! \mathbf{1}_{k \ell} (i) \!\left[ \bm{a}_{\ell k} (i) \bm{\psi}_{\ell,i} \!+\! \left(1-\bm{a}_{\ell k}(i)\right)\bm{\psi}_{k,i} \right] \notag
			\end{align}
			\vspace{-4mm}
			\EndLoop
			
		\end{algorithmic}
	\end{algorithm}	
			
	\section{Stability Analysis and Limiting Behavior}
	\label{stability}
	
	In this section, we study the stability of Algorithm~\ref{alg:summary} and its limiting performance after sufficiently long iterations. In order to pursue a mathematically tractable analysis, we assume that the maximum benefit $\bar{\bm{b}}_k(i)$ is estimated rather accurately by each agent $k$. That is, instead of the real-time implementation (\ref{bapprox})--(\ref{insapp}), we consider (\ref{bhat}) that
	\begin{align}
		\label{bhat2}
		\bar{\bm{b}}_k(i)=\|(I-\mu R_{u,k})  \widetilde{\bm{w}}_{k,i-1}\|^2_{R_{u,k}}
	\end{align}
	This consideration is motivated by taking the expectation of expression (\ref{bapprox}) given $\widetilde{\bm{w}}_{k,i-1}$:
	\begin{align}
		\label{expbapprox}
		&\mathbb{E}\Big[(1-\mu\|\bm{u}_{k,i}\|^2)^2 \widetilde{\bm{w}}_{k,i-1}^* \bm{u}_{k,i}^* \bm{u}_{k,i} \widetilde{\bm{w}}_{k,i-1}\Big| \widetilde{\bm{w}}_{k,i-1}\Big] \notag \\
		&=\widetilde{\bm{w}}_{k,i-1}^* \Big[R_{u,k}-2\mu \mathbb{E}[\bm{u}_{k,i}^* \bm{u}_{k,i}\bm{u}_{k,i}^* \bm{u}_{k,i}]+O(\mu^2)\Big]\widetilde{\bm{w}}_{k,i-1}
	\end{align}
	By subtracting (\ref{bhat2}) from (\ref{expbapprox}), we can see that the difference between (\ref{bhat2}) and (\ref{expbapprox}) is at least in the order of $\mu$:
	\begin{align}
		\widetilde{\bm{w}}_{k,i-1}^* \Big[2\mu (R_{u,k}^2- \mathbb{E}[\bm{u}_{k,i}^* \bm{u}_{k,i}\bm{u}_{k,i}^* \bm{u}_{k,i}])+O(\mu^2)\Big]\widetilde{\bm{w}}_{k,i-1}
	\end{align}
	Therefore, for small enough $\mu$ the expected value of the realization given by (\ref{bapprox}) approaches the true value (\ref{bhat2}). The performance degradation from the real-time implementation error will be illustrated by numerical simulations in Sec.~\ref{simulation}.	 
	
	Under this condition, and for $\mu\ll 1$, we shall argue that the operation of each agent is stable in terms of both the estimation cost \textit{and} the communication cost. Specifically, for the estimation cost, we will provide a condition on the step-size to ensure that the mean-square estimation error of each agent is asymptotically bounded. Using this result, we will further show that the communication cost for each agent $k$, and which is denoted by $J^\text{com}_k$, is upper bounded by a constant value that is unrelated to the transmission cost $c_k$. This result will be in contrast to the case of always cooperative agents where $J^\text{com}_k$ will be seen to increase proportionally with $c_k$. This is because in our case, the probability of cooperation, ${\rm Prob}\{\bm{a}_{k \ell}(i)=1\}$, will be shown to be upper bounded by the ratio $c^o/c_k$ for some constant $c^o$ independent of $c_k$. It will then follow that when the communication becomes expensive (large $c_k$), self-interested agents using the adaptive reputation scheme will become unwilling to cooperate. 
	
	\subsection{Estimation Performance}
	\label{Recursion}
	In conventional stability analysis for diffusion strategies, the combination coefficients are either assumed to be fixed, as done in~\cite{Catt10,Chen12,Chen13,Zhao12,Chen151,Chen152}, or their expectations conditioned on the estimates $\bm{w}_{k,i-1}$ are assumed to be constant, as in~\cite{Zhao15}. These conditions are not applicable to our scenario. When self-interested agents employ the nonlinear threshold-based action policy (\ref{stra2}), the ATC diffusion algorithm (\ref{ATC_A}) and (\ref{ATC2}) ends up involving a combination matrix whose entries are {\em dependent} on the estimates $\bm{w}_{k,i-1}$ (or the errors $\widetilde{\bm{w}}_{k,i-1}$). This fact introduces a new challenging aspect into the analysis of the distributed strategy. In the sequel, and in order to facilitate the stability and mean-square analysis of the learning process, we shall examine the performance of the agents in the network in three operating regions: the far-field region, the near-field region, and a region in between. We will show that the evolution of the estimation errors in these operating regions can be described by the same network error recursion given further ahead in (\ref{maxmum}). Following this conclusion, we will then be able to use (\ref{maxmum}) to provide general statements about
	stability and performance in the three regions.
	
	To begin with, referring to the listing in Algorithm~\ref{alg:summary}, we start by noting that we write down the following error recursions for each agent $k$:
	\begin{align}
		\label{errorRecur1}
		\widetilde{\bm{\psi}}_{k,i} &=(I-\mu \bm{u}_{k,i}^* \bm{u}_{k,i}) \widetilde{\bm{w}}_{k,i-1}-\mu \bm{u}_{k,i}^* \bm{v}_k(i) \\
		\label{errorRecur2}
		\widetilde{\bm{w}}_{k,i} &= \sum\limits_{\ell \in \mathcal{N}_k} \bm{g}_{\ell k}(i) \widetilde{\bm{\psi}}_{\ell,i}
	\end{align}
	where the combination coefficients $\{\bm{g}_{\ell k}(i),~\ell \in \mathcal{N}_k\}$ used in (\ref{errorRecur2}) are defined as follows and add up to one:
	\begin{align}
		\label{glk}
		\bm{g}_{\ell k}(i) &\triangleq
		(1-\alpha_k) \mathbf{1}_{\ell k} (i) \bm{a}_{\ell k}(i) \geq 0 \\
		\label{gkk}
		\bm{g}_{k k}(i) &\triangleq 1- \sum\limits_{\ell \in \mathcal{N}_k \setminus \{k\}} \bm{g}_{\ell k}(i) \geq 0 
	\end{align}
	Note that, in view of the pairing process, at most two of the coefficients $\{\bm{g}_{\ell k}(i)\}$ in (\ref{errorRecur2}) are nonzero in each time instant. The subsequent performance analysis will depend on evaluating the squared weighted norm of $\widetilde{\bm{\psi}}_{k,i}$ in (\ref{errorRecur1}), which is seen to be:
	\begin{align}
		\label{psinorm2}
		\|\widetilde{\bm{\psi}}_{k,i}\|^2_{R_{u,k}} &=
		\|(I-\mu \bm{u}_{k,i}^* \bm{u}_{k,i})\widetilde{\bm{w}}_{k,i-1}\|^2_{R_{u,k}} \notag \\
		&~~~+\mu^2 \|\bm{v}_k(i)\|^2 \text{Tr}(\bm{u}^*_{k,i}\bm{u}_{k,i} R_{u,k}) \notag \\
		&~~~-\mu \widetilde{\bm{w}}^*_{k,i-1}(I-\mu \bm{u}_{k,i}^* \bm{u}_{k,i}) R_{u,k}  \bm{u}_{k,i}^* \bm{v}_k(i) \notag \\
		&~~~-\mu \bm{v}^*_k(i) \bm{u}_{k,i} R_{u,k} (I-\mu \bm{u}_{k,i}^* \bm{u}_{k,i}) \widetilde{\bm{w}}_{k,i-1}
	\end{align}
	Now, from (\ref{errorRecur2}) we can use Jensen's inequality and the convexity of the squared norm to write
	\begin{align}
		\|\widetilde{\bm{w}}_{k,i}\|^2_{R_{u,k}} \leq  \sum\limits_{\ell \in \mathcal{N}_k} \bm{g}_{\ell k}(i) \| \widetilde{\bm{\psi}}_{\ell,i} \|^2_{R_{u,\ell}} 
	\end{align}
	so that, under expectation,
	\begin{align}
		\label{errrecur}
		\mathbb{E}&\|\widetilde{\bm{w}}_{k,i}\|^2_{R_{u,k}} \leq \sum\limits_{\ell \in \mathcal{N}_k} \mathbb{E}\left[\bm{g}_{\ell k}(i) \| \widetilde{\bm{\psi}}_{\ell,i} \|^2_{R_{u,\ell}} \right] 
	\end{align}
	We note that $\bm{g}_{\ell k}(i)$ is a function of the random variables $\{\mathbf{1}_{\ell k} (i), \bm{a}_{\ell k}(i)\}$. 
	The random pairing indicator $\mathbf{1}_{\ell k} (i)$ is independent of $\bm{u}_{\ell,i}$ and $\bm{v}_\ell(i)$. As for $\bm{a}_{\ell k}(i)$, which is determined by $\bar{\bm{b}}_\ell(i)$ and $\bm{\theta}_{k \ell}(i)$, we can see from expressions (\ref{bhat2}) and (\ref{repu2}) that both $\bar{\bm{b}}_\ell(i)$ and $\bm{\theta}_{k \ell}(i)$ only depend on the past data prior to time $i$ and therefore are independent of $\bm{u}_{\ell,i}$ and $\bm{v}_\ell(i)$. Consequently, $\bm{g}_{\ell k}(i)$ is independent of $\bm{u}_{\ell,i}$ and $\bm{v}_\ell(i)$, and we get
	\begin{align}
		\label{gpsi}
		\mathbb{E}&\left[\bm{g}_{\ell k}(i) \| \widetilde{\bm{\psi}}_{\ell,i} \|^2_{R_{u,\ell}} \right] \notag \\
		&= \mathbb{E}\Big[\bm{g}_{\ell k}(i) \Big(
		\|(I-\mu \bm{u}_{\ell,i}^* \bm{u}_{\ell,i})\widetilde{\bm{w}}_{\ell,i-1}\|^2_{R_{u,\ell}} \notag \\
		&~~~~~~~~~~~~~~~~+\mu^2 \|\bm{v}_\ell(i)\|^2 \text{Tr}(\bm{u}^*_{\ell,i}\bm{u}_{\ell,i} R_{u,\ell}) \notag \\
		&~~~~~~~~~~~~~~~~-\mu \widetilde{\bm{w}}^*_{\ell,i-1}(I-\mu \bm{u}_{\ell,i}^* \bm{u}_{\ell,i}) R_{u,\ell}  \bm{u}_{\ell,i}^* \bm{v}_\ell(i) \notag \\
		&~~~~~~~~~~~~~~~~-\mu \bm{v}^*_\ell(i) \bm{u}_{\ell,i} R_{u,\ell} (I-\mu \bm{u}_{\ell,i}^* \bm{u}_{\ell,i}) \widetilde{\bm{w}}_{\ell,i-1} \Big)\Big] \notag \\
		&=\mathbb{E}\Big[\bm{g}_{\ell k}(i) 
		\|(I-\mu \bm{u}_{\ell,i}^* \bm{u}_{\ell,i})\widetilde{\bm{w}}_{\ell,i-1}\|^2_{R_{u,\ell}}\Big] \notag \\
		&~~~+\mu^2 \mathbb{E}\bm{g}_{\ell k}(i) 
		\text{Tr}(R^2_{u,\ell}) \sigma^2_{v,\ell}
	\end{align}
	Using the fact that $\bm{u}_{\ell,i}$ is independent of $\bm{g}_{\ell k}(i)$ and $\widetilde{\bm{w}}_{\ell,i-1}$, we get
	\begin{align}
		\label{EG}
		&\mathbb{E}\Big[\bm{g}_{\ell k}(i) 
		\|(I-\mu \bm{u}_{\ell,i}^* \bm{u}_{\ell,i})\widetilde{\bm{w}}_{\ell,i-1}\|^2_{R_{u,\ell}}\Big] \notag \\
		&=\mathbb{E}\Big[\mathbb{E}\Big[\bm{g}_{\ell k}(i) 
		\|(I-\mu \bm{u}_{\ell,i}^* \bm{u}_{\ell,i})\widetilde{\bm{w}}_{\ell,i-1}\|^2_{R_{u,\ell}}\Big| \bm{g}_{\ell k}(i),\widetilde{\bm{w}}_{\ell,i-1} \Big]\Big] \notag \\
		&=\mathbb{E}\Big[\bm{g}_{\ell k}(i) \widetilde{\bm{w}}^*_{\ell,i-1}\Sigma_\ell \widetilde{\bm{w}}_{\ell,i-1}
		\Big] 
	\end{align}
	where
	\begin{align}
		\label{Sigma}
		\Sigma_\ell &\triangleq \mathbb{E} (I-\mu \bm{u}_{\ell,i}^* \bm{u}_{\ell,i}) R_{u,\ell} (I-\mu \bm{u}_{\ell,i}^* \bm{u}_{\ell,i}) \notag \\
		&=R_{u,\ell}-2\mu R^2_{u,\ell} + \mu^2 \mathbb{E}(\bm{u}_{\ell,i}^*\bm{u}_{\ell,i} R_{u,\ell} \bm{u}_{\ell,i}^*\bm{u}_{\ell,i})
	\end{align}
	If the regression data happens to be circular Gaussian, then a closed-form expression exists for the last fourth-order moment term in (\ref{Sigma}) in terms of $R_{u,\ell}$~\cite{Sayed08}. We will not assume Gaussian data. Instead, we will assume that the fourth-order moment is bounded and that the network is operating in the slow adaptation regime with a sufficiently small step-size so that terms that depend on higher-order powers of $\mu$ can be
	ignored in comparison to other terms. Under this assumption, we replace (\ref{EG}) by:
	\begin{align}
		\label{EG2}
		\mathbb{E}\Big[  \bm{g}_{\ell k}(i) 
		&\|(I-\mu \bm{u}_{\ell,i}^* \bm{u}_{\ell,i})\widetilde{\bm{w}}_{\ell,i-1}\|^2_{R_{u,\ell}}\Big] \notag \\
		&=\mathbb{E}\Big[\bm{g}_{\ell k}(i) \widetilde{\bm{w}}^*_{\ell,i-1}\Omega'_\ell \widetilde{\bm{w}}_{\ell,i-1} \Big] \notag \\
		&=\mathbb{E}\Big[\bm{g}_{\ell k}(i) \bar{\bm{b}}_\ell(i) \Big]
	\end{align}
	where $\Omega'_\ell=R_{u,\ell}(I-\mu R_{u,\ell})^2$ from (\ref{Omega}) and $\bar{\bm{b}}_\ell(i)=\widetilde{\bm{w}}^*_{\ell,i-1}\Omega'_\ell \widetilde{\bm{w}}_{\ell,i-1}$ from (\ref{bhat2}). Note that 
	\begin{align}
		\Sigma_\ell - \Omega_\ell' =  O(\mu^2)
	\end{align}
	Therefore, expression (\ref{gpsi}) becomes
	\begin{align}
		\label{psibound}
		&\mathbb{E}\!\left[\bm{g}_{\ell k}(i) \| \widetilde{\bm{\psi}}_{\ell,i} \|^2_{R_{u,\ell}} \right] \!=\! \mathbb{E}\!\Big[\bm{g}_{\ell k}(i) \bar{\bm{b}}_\ell(i) \!\Big]
		\!\!+\!\mu^2 \text{Tr}(R^2_{u,\ell}) \sigma^2_{v,\ell} \mathbb{E}\bm{g}_{\ell k}(i) 
	\end{align}
	To continue, we introduce the following lemma which provides useful bounds for $\bar{\bm{b}}_k(i)$.
	\vspace{2mm}
	\begin{lemma} [Bounds on $\bar{\bm{b}}_k(i)$] 
		\label{lem:boundb}
		The values of $\bar{\bm{b}}_k(i)$ defined by (\ref{bhat2}) are lower and upper bounded by:
		\begin{align}
			\label{buplowbound}
			\rho^2_\text{min} \|\widetilde{\bm{w}}_{k,i-1}\|^2_{R_{u,k}} \leq \bar{\bm{b}}_k(i) \leq \rho^2_\text{max} \|\widetilde{\bm{w}}_{k,i-1}\|^2_{R_{u,k}} 
		\end{align}
		where
		\begin{align}
			\label{rhomax}
			\rho_\text{max} &\triangleq \max\limits_{1\leq k \leq N} \lambda_\text{max}(I-\mu R_{u,k}) \\
			\label{rhomin}
			\rho_\text{min} &\triangleq \min\limits_{1\leq k \leq N} \lambda_\text{min}(I-\mu R_{u,k})
		\end{align}
	\end{lemma}
	\begin{proof} 
		We introduce the eigendecomposition of the covariance matrix, $R_{u,k} \triangleq U_k \Lambda_k U_k^*$, where $U_k$ is a unitary matrix and $\Lambda_k \triangleq \text{diag}\{\lambda_{1,k},...,\lambda_{M,k}\}$ is a diagonal matrix with positive entries. Then, $R_{u,k}$ can be factored as 
		\begin{align}
			R_{u,k}=R_{u,k}^\frac{1}{2}R_{u,k}^\frac{1}{2}
		\end{align}
		where 
		\begin{align}
			R_{u,k}^\frac{1}{2} &\triangleq U_k \Lambda_k^\frac{1}{2} U_k^*,~~ 
			\Lambda_k^\frac{1}{2} \triangleq \text{diag}\{\sqrt{\lambda_{1,k}},...,\sqrt{\lambda_{M,k}}\}
		\end{align}
		It is easy to verify that $R_{u,k}^\frac{1}{2}$ and $I-\mu R_{u,k}$ are commutable. 
		Using this property, we obtain the following inequality:
		\begin{align}
			\label{bupper}
			&\widetilde{\bm{w}}^*_{k,i-1} \Omega_k' \widetilde{\bm{w}}_{k,i-1} \notag \\
			&~~=\widetilde{\bm{w}}_{k,i-1}^* (I-\mu R_{u,k})R_{u,k}^\frac{1}{2} R_{u,k}^\frac{1}{2} (I-\mu R_{u,k})\widetilde{\bm{w}}_{k,i-1} \notag \\
			&~~=\widetilde{\bm{w}}_{k,i-1}^* R_{u,k}^\frac{1}{2}(I-\mu R_{u,k})^2 R_{u,k}^\frac{1}{2}\widetilde{\bm{w}}_{k,i-1} \notag \\
			&~~\leq \lambda_\text{max}((I-\mu R_{u,k})^2) \widetilde{\bm{w}}_{k,i-1}^* R_{u,k} \widetilde{\bm{w}}_{k,i-1} \notag \\
			&~~\leq \rho^2_\text{max} \|\widetilde{\bm{w}}_{k,i-1}\|^2_{R_{u,k}}
		\end{align}
		We can obtain the lower bound for $\bar{\bm{b}}_k(i)$ by similar arguments.
	\end{proof}
	\vspace{2mm}
	Using the upper bound from Lemma~\ref{lem:boundb}, we have
	\begin{align}
		\label{term1}
		\mathbb{E}\Big[\bm{g}_{\ell k}(i) \bar{\bm{b}}_\ell(i)\Big] \leq \rho_\text{max}^2  \mathbb{E}\Big[\bm{g}_{\ell k}(i) \|\widetilde{\bm{w}}_{\ell,i-1}\|^2_{R_{u,\ell}}
		\Big]
	\end{align}
	Then, from (\ref{psibound}) and (\ref{term1}) we deduce from (\ref{errrecur}) that
	\begin{align}
		\label{singleexpress1}
		\mathbb{E}\|\widetilde{\bm{w}}_{k,i}\|^2_{R_{u,k}} 
		&\leq \rho_\text{max}^2 \! \sum\limits_{\ell \in \mathcal{N}_k} \! \mathbb{E}\Big[\bm{g}_{\ell k}(i) \|\widetilde{\bm{w}}_{\ell,i-1}\|^2_{R_{u,\ell}}
		\Big] \notag \\
		&~~~+\mu^2 \sum\limits_{\ell \in \mathcal{N}_k} \text{Tr}(R^2_{u,\ell}) \sigma^2_{v,\ell} \mathbb{E}\bm{g}_{\ell k}(i)
	\end{align}
	From (\ref{glk})--(\ref{gkk}), it is ready to check that $\{\mathbb{E}\bm{g}_{\ell k}(i)\}$ are nonnegative and add up to 1. Therefore, the second term on the right-hand side of (\ref{singleexpress1}) is a convex combination and has the following upper bound:
	\begin{align}
		\label{term2}
		\mu^2 \sum\limits_{\ell \in \mathcal{N}_k} 
		\text{Tr}(R^2_{u,\ell}) \sigma^2_{v,\ell} \mathbb{E}\bm{g}_{\ell k}(i) \leq \mu^2 \kappa 
	\end{align}
	where 
	\begin{align}
		\kappa \triangleq \max\limits_{1\leq k \leq N} \text{Tr}(R^2_{u,k}) \sigma^2_{v,k}
	\end{align}
	Therefore, we have
	\begin{align}
		\label{singleexpress}
		\mathbb{E}\|\widetilde{\bm{w}}_{k,i}\|^2_{R_{u,k}} 
		&\leq \rho_\text{max}^2 \! \sum\limits_{\ell \in \mathcal{N}_k} \! \mathbb{E}\Big[\bm{g}_{\ell k}(i) \|\widetilde{\bm{w}}_{\ell,i-1}\|^2_{R_{u,\ell}}
		\Big] +\mu^2 \kappa
	\end{align}
	Since the combination coefficients $\{\bm{g}_{\ell k}(i)\}$ and the estimation errors $\{\|\widetilde{\bm{w}}_{k,i-1}\|^2_{R_{u,k}}\}$ are related in a nonlinear manner (as revealed by (\ref{stra2}, (\ref{bhat2}), and (\ref{glk})--(\ref{gkk})), the analysis of Algorithm~\ref{alg:summary} becomes challenging. To continue, we examine the behavior of the agents in the three regions of operation mentioned before.
	
	During the initial stage of adaptation, agents are generally away from the target vector $w^o$ and therefore have large estimation errors. 
	We refer to this domain as the far-field region of operation, and we will characterize it by the condition:
	\begin{align}
		\label{largeest}
		\text{Far-Field:}~~~{\rm Prob}\left\{\|\widetilde{\bm{w}}_{k,i-1}\|^2_{R_{u,k}} > \frac{c_k \chi_k}{\rho_\text{min}^2\varepsilon}\right\} >\phi 
	\end{align}
	where $\chi_k$ and $\rho_\text{min}$ are defined in (\ref{chi}) and (\ref{rhomin}), respectively, and the parameter $\phi$ is close to 1 and in the range of $1 \geq \phi \gg 0$. 
	That is, in the far-field regime, the weighted squared norm of estimation error $\|\widetilde{\bm{w}}_{k,i-1}\|^2_{R_{u,k}}$ exceeds a threshold with high probability. We note that 
	the far-field condition (\ref{largeest}) can be more easily achieved when $c_k$ is small. We also note that when the event in (\ref{largeest}) holds with high-probability, then agent $k$ is enticed to cooperate
	since:
	\begin{align}
		\label{cond_a1}
		\|\widetilde{\bm{w}}_{k,i-1}\|^2_{R_{u,k}} > \frac{c_k \chi_k}{\rho_\text{min}^2 \varepsilon} 
		\stackrel{(a)}{\Rightarrow}~ &\bar{\bm{b}}_k(i) \bm{\theta}_{\ell k}(i) > c_k \chi_k \notag \\
		\stackrel{(b)}{\Rightarrow}~ &\bm{a}_{k \ell}(i)=1   
	\end{align}
	where step $(a)$ is by (\ref{buplowbound}) and the fact $\bm{\theta}_{\ell k}(i)\in [\varepsilon,1]$, and step $(b)$ is by (\ref{stra2}). Consequently, in the far-field region it holds with high likelihood that
	\begin{align}
		\label{coopapprox}
		{\rm Prob}\{\bm{a}_{k \ell}(i)=1\} \approx 1 
	\end{align}
	This approximation explains the phenomenon that with large estimation errors, agents are willing to cooperate and share estimates. We then say that under the far-field condition (\ref{largeest}), the combination coefficients in (\ref{glk})--(\ref{gkk}) can be expressed as
	\begin{align}
		\label{farcomblk}
		\bm{g}_{\ell k}(i) &=
		(1-\alpha_k) \mathbf{1}_{\ell k} (i) \\
		\label{farcombkk}
		\bm{g}_{k k}(i) &= 1- \sum\limits_{\ell \in \mathcal{N}_k \setminus \{k\}} \bm{g}_{\ell k}(i) \notag \\
		&=\alpha_k+(1-\alpha_k) \mathbf{1}_{k k} (i)
	\end{align}
	with expectation values as follows:
	\begin{align}
		g_{\ell k} &\triangleq \mathbb{E} \bm{g}_{\ell k}(i)=
		(1-\alpha_k) p_{\ell k} \notag \\
		g_{k k} &\triangleq \mathbb{E} \bm{g}_{k k}(i) = \alpha_k + (1-\alpha_k) p_{k k}
	\end{align}
	In this case, expression (\ref{singleexpress}) becomes
	\begin{align}
		\label{farfieldexpress}
		\mathbb{E}\|\widetilde{\bm{w}}_{k,i}\|^2_{R_{u,k}} 
		&\leq \rho_\text{max}^2 \sum\limits_{\ell \in \mathcal{N}_k} g_{\ell k}\mathbb{E} \|\widetilde{\bm{w}}_{\ell,i-1}\|^2_{R_{u,\ell}}
		+\mu^2 \kappa
	\end{align}
	where we used the independence between $\mathbf{1}_{\ell k} (i)$ and $\|\widetilde{\bm{w}}_{\ell,i-1}\|^2_{R_{u,\ell}}$. 
	If we stack $\{\|\widetilde{\bm{w}}_{k,i}\|^2_{R_{u,k}}\}$ into a vector $\bm{\mathcal{X}}_i \triangleq \text{col}\{\|\widetilde{\bm{w}}_{1,i}\|^2_{R_{u,1}},...,\|\widetilde{\bm{w}}_{N,i}\|^2_{R_{u,N}}\}$ and collect the combination coefficients $\{g_{\ell k}\}$ into a left-stochastic matrix $G$, then we obtain the network error recursion as:
	\begin{align}
		\label{absrecur}
		\mathbb{E}\bm{\mathcal{X}}_i \preceq \rho_\text{max}^2 G^T  (\mathbb{E} \bm{\mathcal{X}}_{i-1})+\mu^2 \kappa e
	\end{align}
	where $e \in 1_N$ is the vector of all ones and the notation $x \preceq y$ denotes that the components of vector $x$ are less than or equal to the corresponding components of vector $y$.
	Taking the maximum norm from both sides and using the left-stochastic matrix property $\|G^T\|_\infty=1$, we obtain:
	\begin{align}
		\label{maxmum}
		\max_{1\leq k \leq N}\mathbb{E}\|\widetilde{\bm{w}}_{k,i}\|^2_{R_{u,k}} \triangleq \|\mathbb{E}\bm{\mathcal{X}}_i\|_\infty < \rho_\text{max}^2 \|\mathbb{E} \bm{\mathcal{X}}_{i-1}\|_\infty+\mu^2 \kappa 
	\end{align}
	Let us consider the long-term scenario $i \gg 1$.
	The far-field regime (\ref{largeest}) is more likely to occur when $c_k$ is small. Let us now examine the situation in which the communication cost is expensive ($c_k \gg 0$). In this case, the agents will operate in a near-field regime, which we characterize by the condition:
	\begin{align}
		\label{nearfield}
		\text{Near-Field:}~~~{\rm Prob}\left\{\| \widetilde{\bm{w}}_{k,i-1}\|^2_{R_{u,k}} < \frac{c_k \chi_k}{\rho_\text{max}^2}\right\} >\phi 
	\end{align}
	where $\rho_\text{max}$ is defined in (\ref{rhomax}). We note that 
	\begin{align}
		\label{cond_a2}
		\|\widetilde{\bm{w}}_{k,i-1}\|^2_{R_{u,k}} < \frac{c_k \chi_k}{\rho_\text{max}^2} 
		\Rightarrow~ &\bar{\bm{b}}_k(i) \bm{\theta}_{\ell k}(i) < c_k \chi_k \notag \\
		\Rightarrow~ &\bm{a}_{k \ell}(i)=0   
	\end{align}
	In this regime, it then holds with high likelihood that
	\begin{align}
		{\rm Prob}\{\bm{a}_{k \ell}(i)=0\} \approx 1 
	\end{align}
	and the combination coefficients in (\ref{glk})--(\ref{gkk}) then become:
	\begin{align}
		\label{nearcomb}
		\bm{g}_{k k}(i) = 1, ~\bm{g}_{\ell k}(i) = 0
	\end{align}
	This means that agents will now be operating non-cooperatively since the benefit of sharing estimates is small relative to the expensive communication cost $c_k \gg 0$. Using similar arguments to (\ref{farfieldexpress})--(\ref{maxmum}), we arrive at the same network recursion (\ref{maxmum}) for this regime.
	
	However, there exists a third possibility that for moderate values of $c_k$, agents operate at a region that does not belong to neither the far-field nor the near-field regimes. In this region, agents will be choosing to cooperate or not depending on their local conditions. To facilitate the presentation, let us introduce the notation $I_{\mathbb{E}}$ for an indicator function over event $\mathbb{E}$ where $I_{\mathbb{E}}=1$ if event $\mathbb{E}$ occurs and $I_{\mathbb{E}}=0$ otherwise. Then, the action policy (\ref{stra2}) can be rewritten as:
	\begin{align}
		\bm{a}_{k \ell}(i)= I_{\bar{\bm{b}}_k(i) \bm{\theta}_{\ell k}(i)> c_k \chi_k} 
	\end{align}
	From Lemma~\ref{lem:boundb}, we know that $\bar{\bm{b}}_k(i)$ is sandwiched within an interval of width $(\rho^2_\text{max}-\rho^2_\text{min}) \|\widetilde{\bm{w}}_{k,i-1}\|^2_{R_{u,k}}$. As the error $\|\widetilde{\bm{w}}_{k,i-1}\|^2_{R_{u,k}}$ becomes smaller after sufficient iterations, the feasible region of $\bar{\bm{b}}_k(i)$ shrinks. Therefore, it is reasonable to assume that $\bar{\bm{b}}_k(i)$ becomes concentrated around its mean $\bar{b}_k(i) \triangleq \mathbb{E}\bar{\bm{b}}_k(i)$ for $i \gg 0$. It follows that we can approximate $\bm{a}_{k \ell}(i)$ by replacing $\bar{\bm{b}}_k(i)$ by $\bar{b}_k(i)$:
	\begin{align}
		\label{approxaction}
		\bm{a}_{k \ell}(i) \approx I_{\bar{b}_k(i) \bm{\theta}_{\ell k}(i)> c_k \chi_k},~~\text{as}~~i\gg 1
	\end{align}
	To continue, we further assume that after long iterations, which means that repeated interactions between agents have occurred for many times, the reputation scores gradually become stationary, i.e., we can model $\{\bm{\theta}_{\ell k}(i)\}$ as
	\begin{align}
		\bm{\theta}_{\ell k}(i) = \theta_{\ell k} + \bm{n}_{\ell k}(i),~~\text{as}~~i\gg 1
	\end{align}
	where $\theta_{\ell k} \triangleq \mathbb{E}\bm{\theta}_{\ell k}(i)$ and $\bm{n}_{\ell k}(i)$ is a random disturbance which is independent and identically distributed (i.i.d.) over time $i$ and assumed to be independent of all other random variables. Under this modeling, we can obtain the independence between $\bm{n}_{\ell k}(i)$ and $\widetilde{\bm{w}}_{\ell,i-1}$ and write that for $k \neq \ell$
	\begin{align}
		&\mathbb{E}\Big[\bm{g}_{\ell k}(i) \|\widetilde{\bm{w}}_{\ell,i-1}\|^2_{R_{u,\ell}}
		\Big] \notag \\ 
		&= \mathbb{E} \Big[ (1-\alpha_k) \mathbf{1}_{\ell k} (i) \bm{a}_{\ell k}(i) \|\widetilde{\bm{w}}_{\ell,i-1}\|^2_{R_{u,\ell}}\Big] \notag \\
		&\approx (1-\alpha_k) p_{\ell k} \mathbb{E} \Big[I_{\bar{b}_\ell(i) (\theta_{\ell k}+ \bm{n}_{\ell k}(i))> c_k \chi_k} \cdot \|\widetilde{\bm{w}}_{\ell,i-1}\|^2_{R_{u,\ell}}\Big] \notag \\
		&= g'_{\ell k} \mathbb{E} \|\widetilde{\bm{w}}_{\ell,i-1}\|^2_{R_{u,\ell}}
	\end{align}
	where 
	\begin{align}
		g'_{\ell k} &\triangleq (1-\alpha_k) p_{\ell k} \mathbb{E} \big[ I_{\bar{b}_k(i) (\theta_{\ell k}+ \bm{n}_{\ell k}(i))> c_k \chi_k}\big] \geq 0
	\end{align}
	For $\ell = k$, we can use similar arguments to write
	\begin{align}
		&\mathbb{E}\Big[\bm{g}_{k k}(i) \|\widetilde{\bm{w}}_{k,i-1}\|^2_{R_{u,k}}
		\Big] \notag \\ 
		&= \mathbb{E} \Big[\Big( 1- \sum\limits_{\ell \in \mathcal{N}_k \setminus \{k\}} (1-\alpha_k) \mathbf{1}_{\ell k} (i) \bm{a}_{\ell k}(i)\Big) \|\widetilde{\bm{w}}_{k,i-1}\|^2_{R_{u,k}}\Big] \notag \\
		&\approx g'_{\ell k} \mathbb{E} \|\widetilde{\bm{w}}_{k,i-1}\|^2_{R_{u,\ell}}
	\end{align}
	where 
	\begin{align}
		g'_{k k} &\triangleq 1-\sum\limits_{\ell \in \mathcal{N}_k \setminus \{k\}} g'_{\ell k} \geq 0  
	\end{align}
	Therefore, expression (\ref{singleexpress}) becomes
	\begin{align}
		\label{midfieldexpress}
		\mathbb{E}\|\widetilde{\bm{w}}_{k,i}\|^2_{R_{u,k}} 
		&\leq \rho_\text{max}^2 \sum\limits_{\ell \in \mathcal{N}_k} g'_{\ell k}\mathbb{E} \|\widetilde{\bm{w}}_{\ell,i-1}\|^2_{R_{u,\ell}}
		+\mu^2 \kappa
	\end{align}
	Following similar arguments to (\ref{farfieldexpress})--(\ref{maxmum}), we again arrive at the same network recursion (\ref{maxmum}). 
	
	We therefore conclude that after long iterations, the estimation performance can be approximated by recursion (\ref{maxmum}) for general values of $c_k$. Consequently, sufficiently small step-sizes that satisfy the following condition guarantees the stability of the network error $\|\mathbb{E}\bm{\mathcal{X}}_i\|_\infty$: 
	\begin{align}
		\label{farcond}
		\rho_\text{max}^2<1 \Leftrightarrow \mu < \frac{2}{\max\limits_{1\leq k \leq N} \lambda_\text{max}(R_{u,k})}
	\end{align}
	which leads to
	\begin{align}
		\label{supsmallc}
		\limsup\limits_{i \rightarrow \infty} \|\mathbb{E}\bm{\mathcal{X}}_i\|_\infty \leq \frac{\mu^2 \kappa }{1-\rho^2_\text{max}}
	\end{align}
	Recalling that we are assuming sufficiently small $\mu$, we have 
	\begin{align}
		\rho_\text{max}^2 
		&=\max_{1\leq k \leq N} \max_{1\leq m \leq M} (1-\mu \lambda_m(R_{u,k}))^2 \notag \\
		&\approx \max_{1\leq k \leq N} \max_{1\leq m \leq M} 1-2\mu \lambda_m(R_{u,k}) \notag \\
		&\triangleq 1-2 \mu \beta
	\end{align}
	where
	\begin{align}
		\label{beta}
		\beta &\triangleq \min_{1\leq k \leq N} \lambda_\text{min}(R_{u,k})
	\end{align}
	It is straightforward to verify that the bound on the right-hand side of (\ref{supsmallc}) is upper-bounded by $\mu \kappa /2\beta$, which is $O(\mu)$. Consequently, we conclude that 
	\begin{align}
		\label{limsup}
		\limsup\limits_{i \rightarrow \infty} \mathbb{E} \|\widetilde{\bm{w}}_{k,i}\|^2_{R_{u,k}}=O(\mu)
	\end{align}
	which establishes the mean-square stability of the network under the assumed conditions and for small
	step-sizes.	
			
	\subsection{Expected Individual and Public Cost}
	
	In this section we assess the probability that agents will opt for cooperation after sufficient time has elapsed and the network has become stable. 
	From (\ref{limsup}), we know that after sufficient iterations, the MSE cost at each agent $k$ is bounded, say, as
	\begin{align}
		\label{indubound}
		\mathbb{E} \|\widetilde{\bm{w}}_{k,i} \|^2_{R_{u,k}} \leq \eta \mu,~~\text{for}~~i \gg 1
	\end{align}
	for some constant $\eta$. 
	Based on this result, we can upper bound the cooperation rate of every agent $k$, defined as the probability that agent $k$ would select $\bm{a}_{k\ell}(i)=1$ for every pairing agent $\ell$.
	
	\begin{thm} (\textit{Upper bound on cooperation rate}) 
		\label{thm:uppbound}
		After sufficient iterations, the cooperation rate for each agent $k$ is upper bounded by:
		\begin{align}
			\label{thm2}
			{\rm Prob}\{\bm{a}_{k \ell}(i)=1\}
			\leq \min\left\{\frac{c^o}{c_k},1\right\}  ,~~\text{for any}~~c_k < \infty
		\end{align}
		where $c^o$ is independent of $c_k$ and defined as
		\begin{align}
			c^o \triangleq \frac{ \eta \mu \rho^2_\text{max} }{\chi_\text{min}},~~\chi_\text{min} \triangleq \min_{1\leq k \leq N} \chi_k
		\end{align}
	\end{thm}
	\begin{proof} 
		From (\ref{stra2}), the cooperation rate of agent $k$ is bounded by:
		\begin{align}
			\label{coopbound}
			{\rm Prob}\{\bm{a}_{k \ell}(i)=1\} &={\rm Prob} \left\{\bar{\bm{b}}_k(i) \bm{\theta}_{\ell k}(i) > c_k \chi_k\right\} \notag \\
			&\leq {\rm Prob} \left\{\bar{\bm{b}}_k(i) > c_k \chi_k \right\} \notag \\
			&\leq {\rm Prob} \Bigg\{ \| \widetilde{\bm{w}}_{k,i-1}\|^2_{R_{u,k}}> \frac{c_k \chi_k}{\rho^2_\text{max}}\Bigg\}
		\end{align}
		where we used the fact that $\bm{\theta}_{k\ell}(i) \leq 1$ and the upper bound on $\bar{\bm{b}}_k(i)$ from (\ref{buplowbound}). Since $\| \widetilde{\bm{w}}_{k,i-1}\|^2_{R_{u,k}}$ is a nonnegative random variable with $\mathbb{E} \|\widetilde{\bm{w}}_{k,i-1} \|^2_{R_{u,k}} \leq \eta \mu$, we can use Markov's inequality~\cite{Papoulis02} to write 
		\begin{align}
			\label{bbound}
			{\rm Prob} &\left\{\|\widetilde{\bm{w}}_{k,i-1}\|^2_{R_{u,k}} > \frac{c_k \chi_k}{\rho^2_\text{max}}\right\} \leq 
			\frac{ \eta \mu \rho^2_\text{max} }{c_k \chi_k} \leq \frac{c^o}{c_k}
		\end{align}
		Combining (\ref{coopbound}) and (\ref{bbound}), we obtain that the cooperation rate for $i\gg 1$ is upper bounded by $c^o/c_k$. Using the fact that ${\rm Prob}\{\bm{a}_{k \ell}(i)=1\}\leq 1$, we get (\ref{thm2}).
	\end{proof}
	\vspace{2mm}
	\begin{figure}[!t]
		\centering
		\includegraphics[width=2.3in]{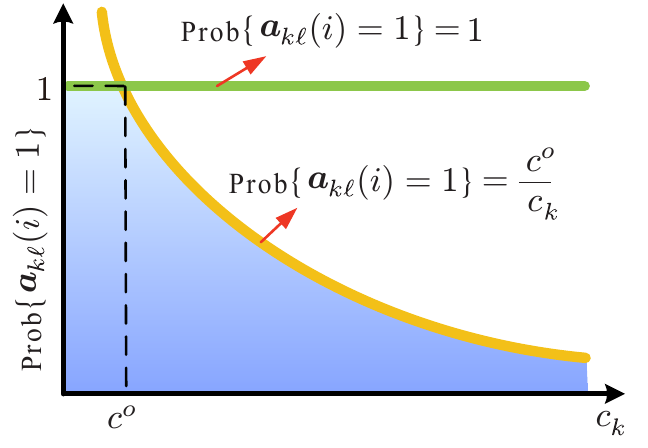}
		\caption{The feasible region of the probability of cooperation ${\rm Prob}\{\bm{a}_{k \ell}(i)=1\}$ for agent $k$.}
		\label{feasible}
	\end{figure}
	
	As illustrated in Fig.~\ref{feasible}, the feasible region of ${\rm Prob}\{\bm{a}_{k \ell}(i)=1\}$ is the intersection area of (\ref{thm2}) and $0\leq {\rm Prob}\{\bm{a}_{k \ell}(i)=1\}\leq 1$. 
	We note that $c^o$ has an order of $\mu$. It describes the fact that when $\mu$ is small, the long term estimation errors reduce and agents have less willingness to cooperate and thus the cooperation rate ${\rm Prob}\{\bm{a}_{k \ell}(i)=1\}$ becomes low.
	Now, the expected communication cost for each agent $k$ is
	\begin{align}
		\mathbb{E}J^\text{com}_k &=\sum_{\ell \in \mathcal{N}_k \setminus \{k\}} \mathbb{E} [\mathbf{1}_{k \ell}(i)\bm{a}_{k \ell}(i)]  c_k \notag \\
		&=\sum_{\ell \in \mathcal{N}_k \setminus \{k\}} p_{k\ell} \cdot {\rm Prob}\{\bm{a}_{k \ell}(i)=1\}  c_k
	\end{align}
	where $\{k\}$ is excluded from the summation since there is no communication cost required for using own data. From Theorem~\ref{thm:uppbound}, we know that when $c_k$ is large, the expected communication cost has an upper bound which is independent of $c_k$, i.e., for $c_k \geq c^o$,
	\begin{align}
		\mathbb{E}J^\text{com}_k &\leq\sum_{\ell \in \mathcal{N}_k \setminus \{k\}} p_{k\ell} \frac{c^o}{c_k} c_k = (1- p_{kk}) c^o
	\end{align}
	On the other hand, the expected estimation cost for each agent $k$ for $i \gg 0$ is:
	\begin{align}
		\mathbb{E}J^\text{est}_k&=\mathbb{E} | \bm{d}_k(i) - \bm{u}_{k,i} \bm{w}_{k,i-1} |^2 \notag \\
		&=\mathbb{E} \| \widetilde{\bm{w}}_{k,i-1} \|^2_{R_{u,k}}+\sigma^2_{v,k} \notag \\
		&\leq \eta \mu +\sigma^2_{v,k}
	\end{align}
	where we use (\ref{qualcost1}) and the fact that $\widetilde{\bm{w}}_{k,i-1}$ is independent of $\{\bm{d}_k(i),\bm{u}_{k,i}\}$.
	It follows that for $i \gg 1$ the expected extended cost at each agent $k$ is bounded by 
	\begin{align}
		\label{uppic}
		\mathbb{E} \left[ J_k (\bm{w}_{k,i-1}, \bm{a}_{k \ell}(i)) \right] &= \mathbb{E}J^\text{est}_k + \mathbb{E}J^\text{com}_k \notag \\
		&\leq \eta \mu +\sigma^2_{v,k} + (1- p_{kk}) c^o
	\end{align}
	If we now define the public cost as the accumulated expected extended cost over the network: 
	\begin{align}
		\label{glob}
		J^\text{pub} &\triangleq \sum_{k=1}^{N} \mathbb{E}\left[ J_k (\bm{w}_{k,i-1}, \bm{a}_{k \ell}(i)) \right] 
	\end{align}
	then
	\begin{align}
		\label{uppgrep}
		J^\text{pub} \leq N \eta \mu     +\sum\limits_{k=1}^{N}\sigma^2_{v,k}+
		c^o \sum\limits_{k=1}^{N}(1- p_{kk})
	\end{align}
	which shows that $J^\text{pub}$ is uniformly bounded by a constant value independent of $c_k$.
	
	For comparison purposes, let us consider a uniform transmission cost $c_k=c$ and consider the case in (\ref{ATC_C}) where the agents are controlled so that they always share estimates with their paired agents, i.e., $\bm{a}_{k \ell}(i)=1$ for all $k$, $\ell$, and $i$ whenever $\mathbf{1}_{k\ell}(i)=1$. Then, the random combination coefficients are the same as (\ref{farcomblk})--(\ref{farcombkk}) defined in the far-field since agents always choose to cooperate. Let us denote the network mean-square-deviation (MSD) of cooperative agents by 
	\begin{align}
		\text{MSD}_\text{coop} &= \lim\limits_{i \rightarrow \infty} \frac{1}{N} \sum_{k=1}^N \mathbb{E} \| \widetilde{\bm{w}}_{k,i} \|^2
	\end{align}
	which has a closed-form expression provided in~\cite{Zhao15}. Therefore, we can characterize the performance of cooperative agents by $\text{MSD}_\text{coop}$ and note that for $i \gg 0$, we have
	\begin{align}
		\sum_{k=1}^N \mathbb{E} \| \widetilde{\bm{w}}_{k,i} \|^2_{R_{u,k}} \geq N \beta \cdot
		\text{MSD}_\text{coop}
	\end{align}
	where $\beta$ is defined in (\ref{beta}). Consequently, after sufficient iterations, the expected public cost for cooperative agents becomes
	\begin{align}
		\label{lpgobe}
		J_\text{coop}^\text{pub} 
		&=\! \sum\limits_{k=1}^{N} \left[ \mathbb{E} \| \widetilde{\bm{w}}_{k,i} \|^2_{R_{u,k}}+\sigma^2_{v,k}+ c  \sum_{\ell \in \mathcal{N}_k \setminus \{k\}} \mathbb{E} \mathbf{1}_{k \ell}(i) \right] \notag \\
		&\geq\! N \beta \cdot \text{MSD}_\text{coop}
		\!+\!\sum\limits_{k=1}^{N} \sigma^2_{v,k}\!+\! c \sum\limits_{k=1}^{N} (1-p_{kk}) 
	\end{align}
	Comparing (\ref{uppgrep}) with (\ref{lpgobe}), we get $J_\text{coop}^\text{pub} \geq J^\text{pub}$ whenever
	\begin{align}
		&~N \beta \cdot \text{MSD}_\text{coop}
		+ c \sum\limits_{k=1}^{N} (1-p_{kk}) \geq N\eta \mu +
		c^o \sum\limits_{k=1}^{N}(1- p_{kk})  \notag \\
		\Leftrightarrow &~ c \geq c^o+ \frac{N\left[
			\eta \mu - \beta \cdot \text{MSD}_\text{coop} \right]}{\sum\limits_{k=1}^{N} (1-p_{kk})}
	\end{align}
	In other words, when the transmission cost $c$ exceeds the above threshold, self-interested agents using the reputation protocol obtain a lower expected public cost than cooperative agents running the cooperative ATC strategy. 
	
	\section{Numerical Results}
	\label{simulation}
	
	We consider a network with $N=20$ agents. The network topology is shown in Fig.~\ref{topology}. The noise variance profile at the agents is shown in Fig.~\ref{noise}. In the simulations, we consider that the transmission cost $c_k=c$ is uniform and the matrix $R_{u,k}=R_u$ is uniform and diagonal.  Figure~\ref{w0ru} shows the entries of the target vector $w^o$ of size $M=10$ and the diagonal entries of $R_u$. We set the step-size at $\mu=0.01$ and the combination weight at $\alpha_k=0.5$ for all $k$. The parameters used in the reputation update rule are $\varepsilon=0.1$ and the initial reputation scores $\bm{\theta}_{k,\ell}(0)=1$ for all agents $k$ and $\ell$. The discount factor $\delta_k$ and the smoothing factor $r_k$ for all $k$ are set to $0.99$ and $0.95$, respectively, and the forgetting factor $\nu$ is set to $0.01$. 
	
	In this simulation, we consider a distributed random-pairing mechanism as follows. At each time instant, each agent independently and uniformly generates a random continuous value from $[0,1]$. Then, the agent holding the smallest value in the network, say agent $k$, is paired with the neighboring agent in $\mathcal{N}_k$ who has not been paired and holds the smallest value in $\mathcal{N}_k \setminus \{k\}$. Then, similar steps are followed by the agent who has not been paired and holds the second smallest value in the network. The random-pairing procedure continues until all agents complete these steps.
	
	In Fig.~\ref{simug}, we simulate the learning curves of instantaneous public costs for small and large communication costs. It is seen that in both cases, using the proposed reputation protocol, the network of self-interested agents reaches the lowest public cost. Therefore, the network is efficient in terms of public cost. Furthermore, we note that in these two cases, there is only small difference between the performance of the reputation protocol using Algorithm~\ref{alg:summary} and the real-time implementation. To see the general effect of $c$, in Fig.~\ref{gencost} we simulate the public cost in steady-state versus the communication cost $c$. We observe that for large and small $c$, the reputation protocol performs as well as the non-sharing and the cooperative network, respectively. The only imperfection occurs around the switching region.
	Without real-time implementations, the reputation protocol has a small degradation in the range of $c \in [10^{-2},10^{-1}]$. While using real-time implementations (\ref{bapprox})--(\ref{insapp}), we can see that the degradation happens in a wider range of $c \in [10^{-4},10^{-1}]$.
	In Fig.~\ref{avgbenefit}, we simulate the network benefit defined as the largest achievable $\bar{\bm{b}}_k(i)$ averaged over all agents in steady state, i.e., 
	\begin{align}
		b^\text{net} \triangleq \lim\limits_{i \rightarrow \infty}\frac{1}{N}\sum_{k=1}^{N} \mathbb{E} [\bar{\bm{b}}_k(i)]
	\end{align}
	where $\bar{\bm{b}}_k(i)$ follows the real-time implementation in (\ref{bapprox})--(\ref{insapp}). We note that the network benefits, $b^\text{net}$, for the non-sharing and cooperative cases are invariant for different communication costs since the behavior of agents is independent of $c$. Moreover, as expected in the form of (\ref{bapprox}), cooperative networks generally give smaller steady-state estimation errors and thus result in lower $b^\text{net}$. When the communication cost $c$ increases, self-interested agents following the proposed reputation protocol have less willingness to cooperate, and therefore have larger estimation errors and predict higher values of the benefit $\bar{\bm{b}}_k(i)$ in general.

	\begin{figure}[!t]
		\centering
		\includegraphics[width=2.5in]{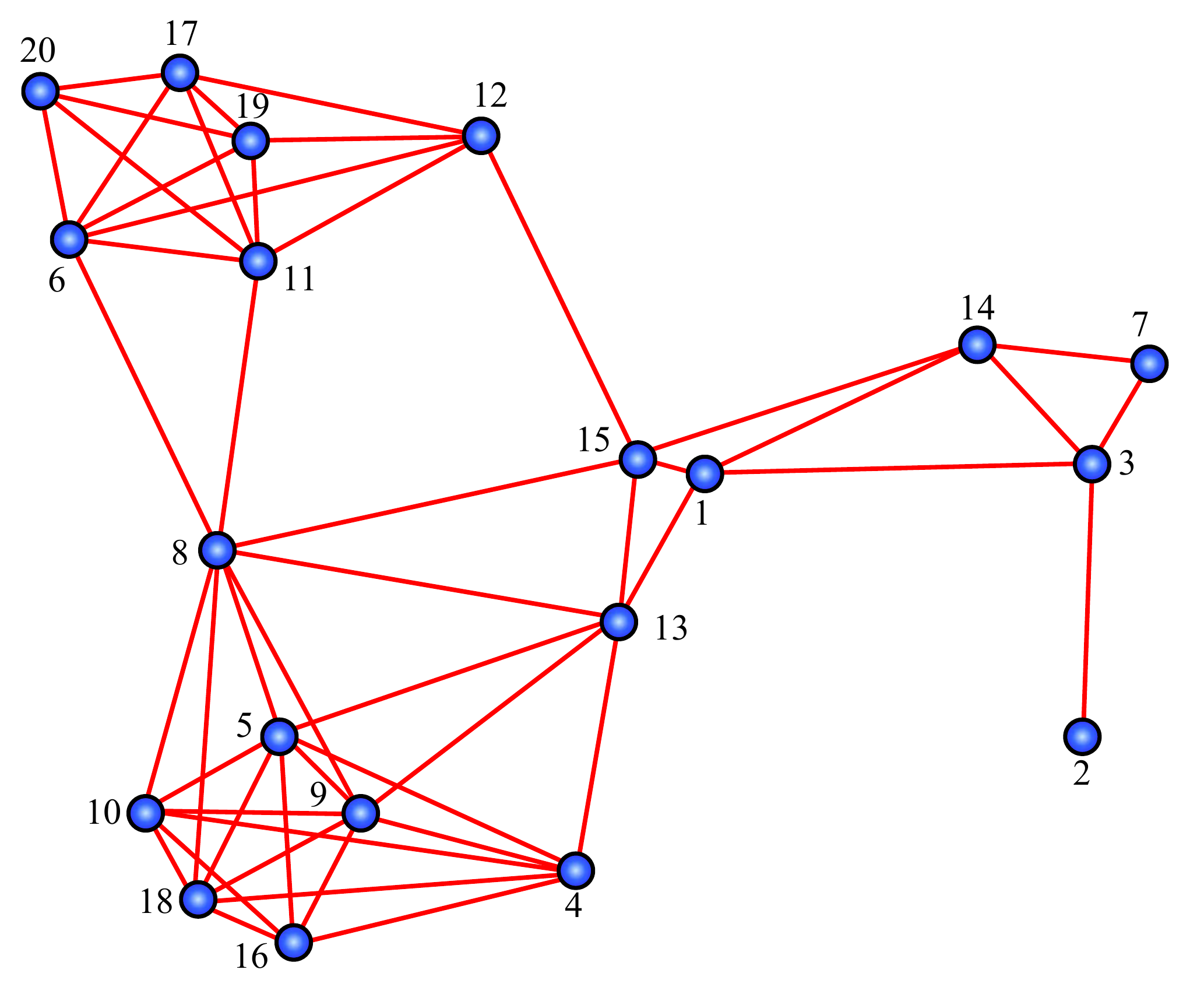}
		\caption{A network topology with $N=20$ agents.}
		\label{topology}
	\end{figure}
	
	\begin{figure}[!t]
		\centering
		\includegraphics[width=3.5in]{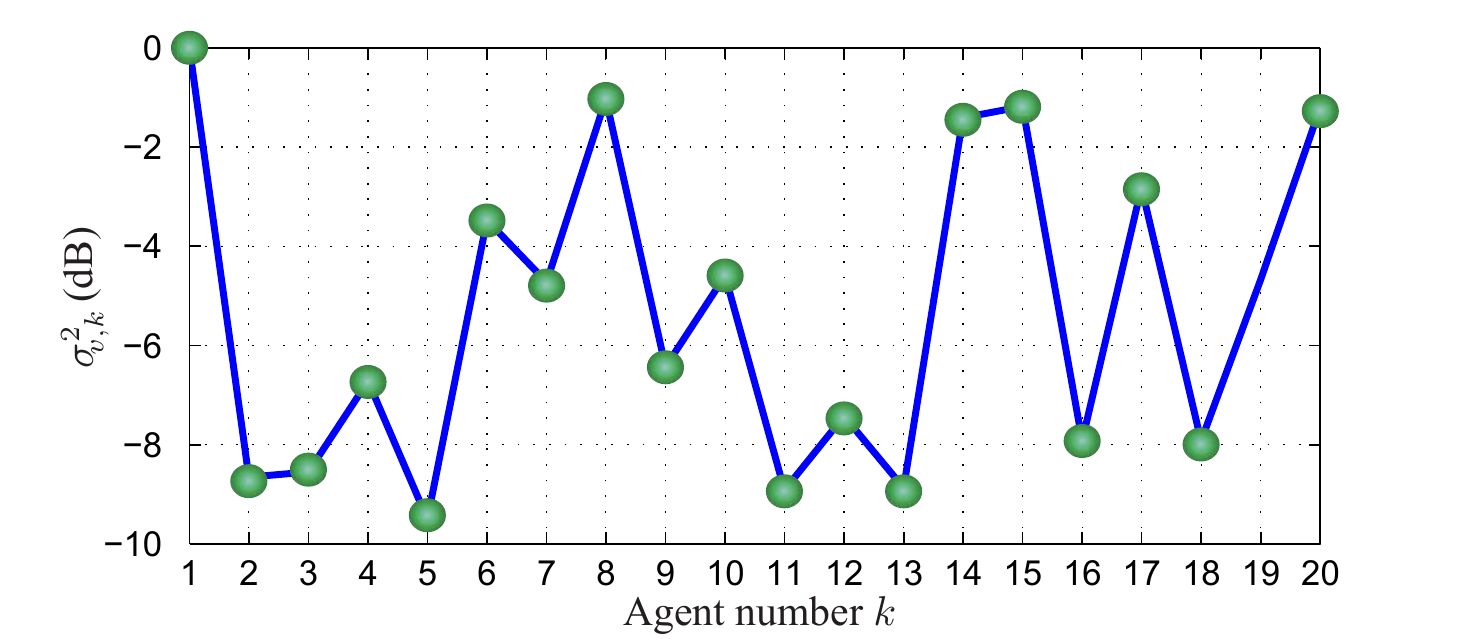}
		\caption{The noise variance profile used in the simulations.}
		\label{noise}
	\end{figure}
	
	\begin{figure}[!t]
		\centering
		\includegraphics[width=3.7in]{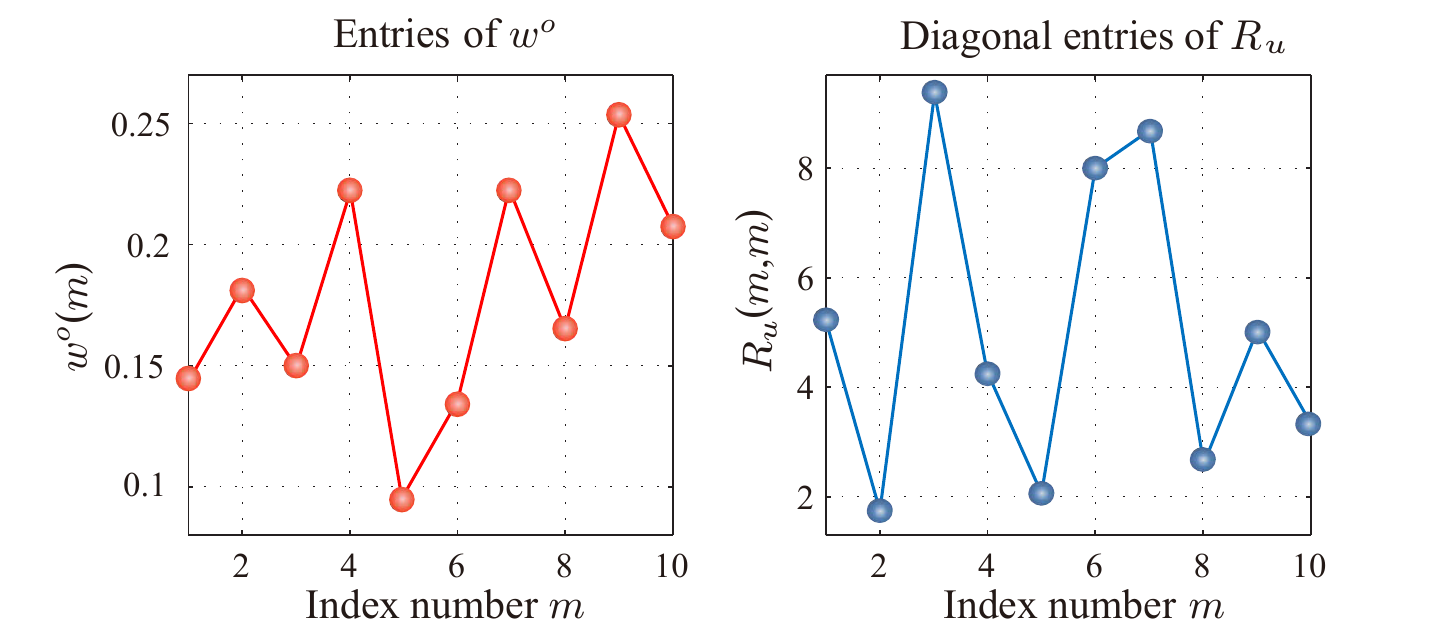}
		\caption{Entries of $w^o$ and $R_u$ used in the simulations.}
		\label{w0ru}
	\end{figure}
	
	\begin{figure}[!t]
		\vspace{-5mm}
		\centering
		\begin{tabular}{c}
			\subfloat[Small communication cost $c=0.0001$.]
			{\includegraphics[width=3.0in]{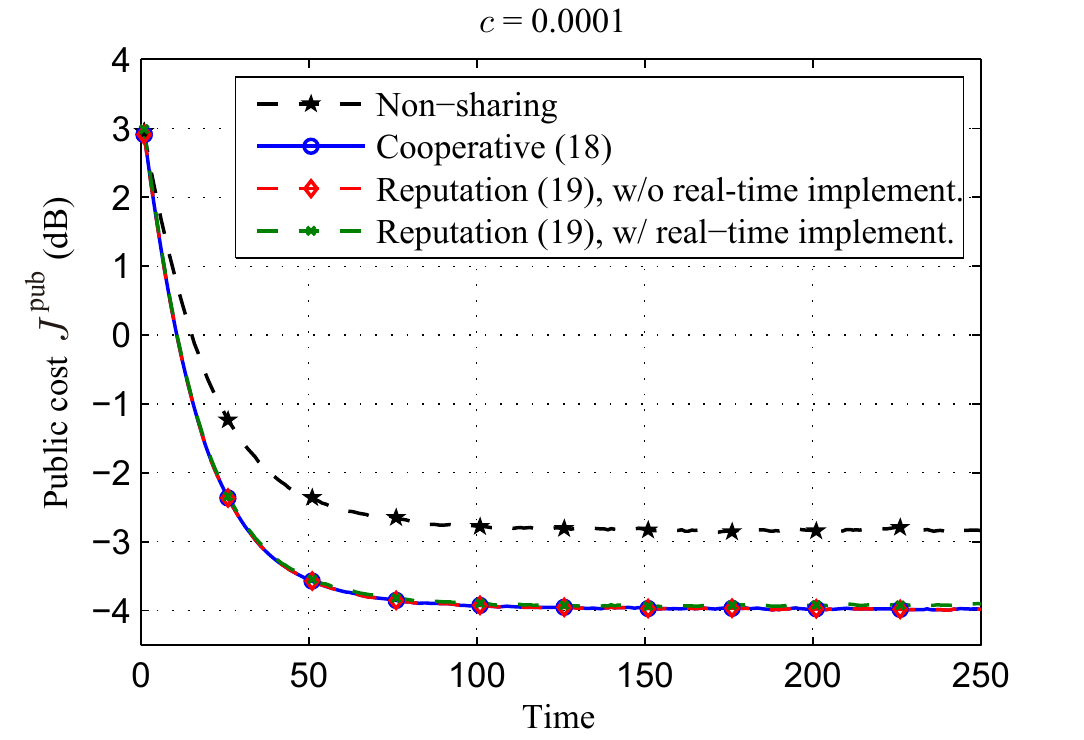}} \\
			\subfloat[Large communication cost $c=0.5$.]
			{\includegraphics[width=3.0in]{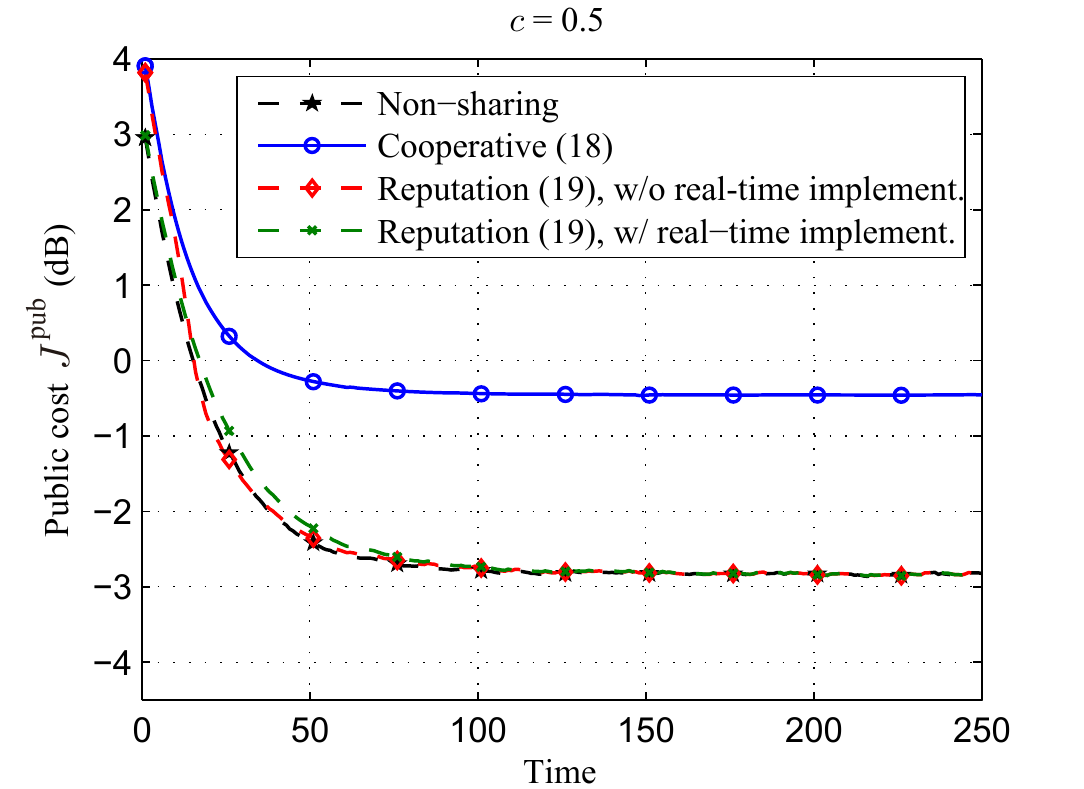}}
		\end{tabular}
		\caption{Learning curve of public cost $J^\text{pub}$ for small and large communication costs.}
		\label{simug}
	\end{figure}
	
	\begin{figure}[t]
		\centering
		\includegraphics[width=3.1in]{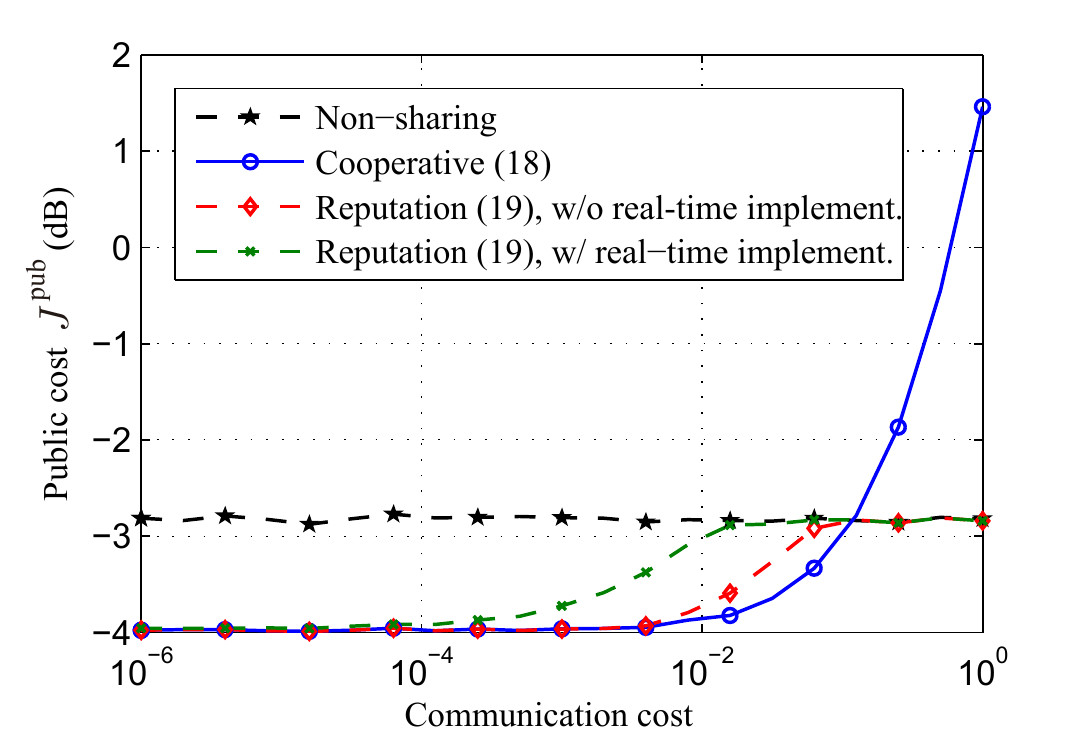}
		\vspace{-2mm}
		\caption{Simulations of steady-state public costs $J^\text{pub}$.}
		\label{gencost}
	\end{figure}
	
	\begin{figure}[!t]
		\vspace{-3mm}
		\centering
		\includegraphics[width=3.1in]{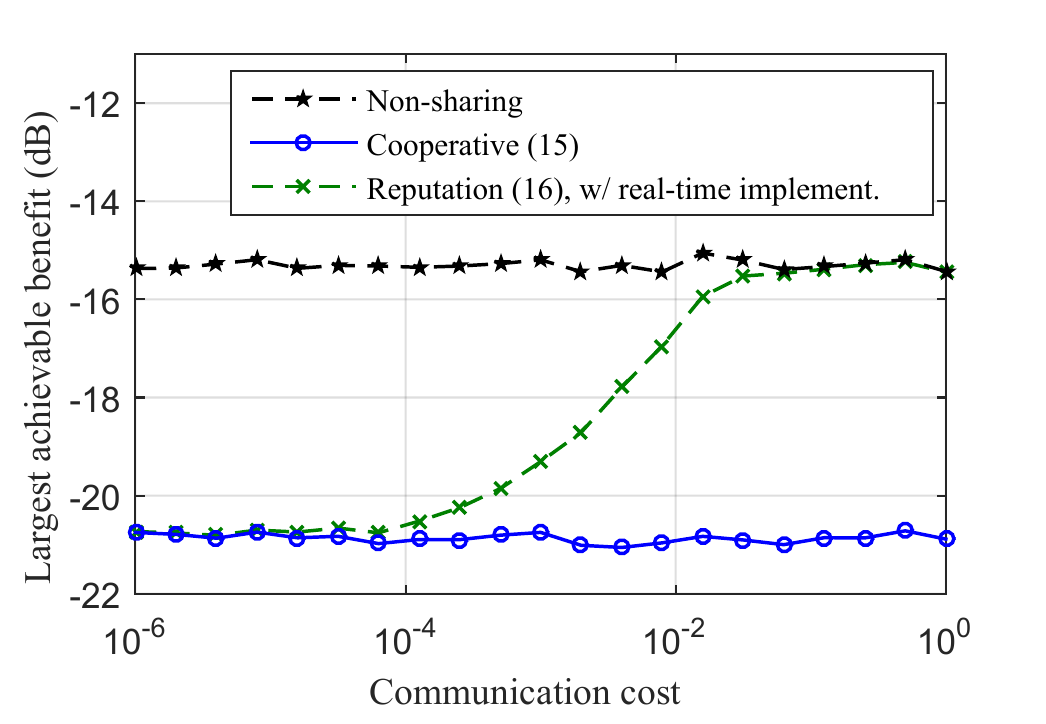}
		\vspace{-2mm}
		\caption{Simulations of average largest achievable benefit $\bar{\bm{b}}_k(i)$ over agents.}
		\label{avgbenefit}
	\end{figure}						
	
	\section{Conclusion}
	In this work, we studied the distributed information-sharing network in the presence of self-interested agents. We showed that without using any historical information to predict future actions, self-interested agents with bounded rationality become non-cooperative and refuse to share information. To entice them to cooperate, we developed an adaptive reputation protocol which turns the best response of self-interested agents into an action-choosing rule.		
			
	\appendices
	\section{}
	
	We can represent the best response rule (\ref{best}) by some mapping function, $f_k(\cdot)$, that maps the available realizations of $\bm{\theta}_{\ell k}(i)$ and $\bm{w}_{k,i-1}$ to $a_{k \ell}(i)$, i.e.,
	\begin{align}
		a_{k \ell}(i)=f_k(\theta_{\ell k}(i),w_{k,i-1})
	\end{align}
	We show the form of the resulting $f_k(\cdot)$ later in (\ref{stra}). For now, we note that construction (\ref{best}) requires us to find the condition for 
	\begin{align}
		J_{k,i}^{\infty' }[a_{k\ell}(i)=1|\bm{w}_{k,i-1}]<J_{k,i}^{\infty' }[a_{k\ell}(i)=0|\bm{w}_{k,i-1}]
	\end{align}
	Using (\ref{Pk}), this condition translates into requiring
	\begin{align}
		\label{cond}
		&J_{k,i}^{\infty' }[a_{k\ell}(i)=1|\bm{w}_{k,i-1}] - J_{k,i}^{\infty' }[a_{k\ell}(i)=0|\bm{w}_{k,i-1}] \notag \\
		&~~~=c_k+\sum_{t=i+1}^{\infty} \delta_k^{t-i} \triangle J_k^1(t)+\sum_{t=i+1}^{\infty} \delta_k^{t-i}\triangle J_k^2(t)<0
	\end{align}
	where, for simplicity, we introduced
	\begin{align}
		\triangle J^1_k(t) &\triangleq  \mathbb{E} \left[
		\bm{a}_{k \ell}(t) c_k \Big{|} 
		\bm{w}_{k,i-1}, \bm{a}_{k \ell}(i)=1, \mathbb{K}_i \right]\notag \\
		&~~~-\mathbb{E} \left[
		\bm{a}_{k \ell}(t) c_k \Big{|} 
		\bm{w}_{k,i-1}, \bm{a}_{k \ell}(i)=0, \mathbb{K}_i \right] \\
		\label{deltaJ2}
		\triangle J^2_k(t) &\triangleq  \mathbb{E} \left[
		J_{k}^\text{act}(\bm{a}_{\ell k}(t)) \Big{|} 
		\bm{w}_{k,i-1}, \bm{a}_{k \ell}(i)=1, \mathbb{K}_i \right] \notag \\
		&~~~- \mathbb{E} \left[
		J_{k}^\text{act}(\bm{a}_{\ell k}(t)) \Big{|} 
		\bm{w}_{k,i-1}, \bm{a}_{k \ell}(i)=0, \mathbb{K}_i \right]
	\end{align}
	Following similar argument to (\ref{prison1}), we combine Assumptions \ref{ass:rationalAgents} and \ref{ass:rationalAgentsEx} to conclude that 
	\begin{align}
		\label{futureaction}
		\bm{a}_{k \ell}(t) = a_{k \ell}(i), ~~\text{for}~t\geq i
	\end{align}
	so that
	\begin{align}
		\label{j1}
		\triangle J^1_k(t)=c_k
	\end{align}
	Similarly, using the assumption $\bm{w}_{k,t-1}=\bm{w}_{k,i-1}$ for $t\geq i$ from (\ref{assump1}), we have
	\begin{align}
		\label{temp}
		\mathbb{E} &\left[J_{k}^\text{act}(\bm{a}_{\ell k}(t)) \Big| \bm{w}_{k,i-1}, \bm{a}_{k \ell}(i)=j, \mathbb{K}_i \right] \notag \\
		&~~~~~~=\mathbb{E} \left[J_{k}^\text{act}(\bm{a}_{\ell k}(t)=1) \Big| \bm{w}_{k,t-1}, \bm{a}_{k \ell}(i)=j, \mathbb{K}_i \right] \notag \\
		&~~~~~~~~~\times B(\bm{a}_{\ell k}(t)=1) \notag \\
		&~~~~~~~~~+\mathbb{E} \left[J_{k}^\text{act}(\bm{a}_{\ell k}(t)=0) \Big| \bm{w}_{k,t-1}, \bm{a}_{k \ell}(i)=j, \mathbb{K}_i \right]  \notag \\
		&~~~~~~~~~\times \left(1-B(\bm{a}_{\ell k}(t)=1)\right)
	\end{align}
	for $j=0$ or 1. From (\ref{BPred}) and (\ref{BPred1}), we can write
	\begin{align}
		\mathbb{E} &\left[J_{k}^\text{act}(\bm{a}_{\ell k}(t)=0) \Big| \bm{w}_{k,t-1}, \bm{a}_{k \ell}(i)=j, \mathbb{K}_i \right] \notag \\
		&~~~~~~~~~~~~~~~~~~~~~~~~~= \mathbb{E} \left[J_{k}^\text{act}(\bm{a}_{\ell k}(t)=0) \Big| \bm{w}_{k,t-1} \right] \notag \\
		\mathbb{E} &\left[J_{k}^\text{act}(\bm{a}_{\ell k}(t)=1) \Big| \bm{w}_{k,t-1}, \bm{a}_{k \ell}(i)=j, \mathbb{K}_i \right] \notag \\
		&~~~~~~~~~~~~~~~~~~~~~~~~~= \mathbb{E} \left[J_{k}^\text{act}(\bm{a}_{\ell k}(t)=1) \Big| \bm{w}_{k,t-1} \right] \notag
	\end{align}
	since these two conditional expectations depend only on $\bm{w}_{k,t-1}$.
	Therefore, expression (\ref{temp}) becomes:
	\begin{align}
		\mathbb{E} &\left[J_{k}^\text{act}(\bm{a}_{\ell k}(t)) \Big| \bm{w}_{k,i-1}, \bm{a}_{k \ell}(i)=j, \mathbb{K}_i \right] \notag \\
		&~~~=\mathbb{E} \left[J_{k}^\text{act}(\bm{a}_{\ell k}(t)=0) \Big| \bm{w}_{k,t-1}\right] \notag \\
		&~~~~~~+B(\bm{a}_{\ell k}(t)=1) \cdot \Big(\mathbb{E} \left[J_{k}^\text{act}(\bm{a}_{\ell k}(t)=1) \Big| \bm{w}_{k,t-1}\right]
		\notag \\
		&~~~~~~~~~~~~~~~~~~~~~~~~~~~~~~~-\mathbb{E} \left[J_{k}^\text{act}(\bm{a}_{\ell k}(t)=0) \Big| \bm{w}_{k,t-1}\right]\Big) \notag \\
		&~~~=\mathbb{E} \left[J_{k}^\text{act}(\bm{a}_{\ell k}(t)=0) \Big| \bm{w}_{k,t-1}\right]- B(\bm{a}_{\ell k}(t)=1)
		\bm{b}_k(t) \notag
	\end{align}
	where we used the definition for $\bm{b}_k(t)$ from (\ref{benefit}). 
	We note that using (\ref{benefit2}) we have $\bm{b}_k(t)=\bm{b}_k(i)$ due to the assumption $\bm{w}_{k,t-1}=\bm{w}_{k,i-1}$. As a result, it follows that 
	\begin{align}
		\mathbb{E} &\left[J_{k}^\text{act}(\bm{a}_{\ell k}(t)) \Big| \bm{w}_{k,i-1}, \bm{a}_{k \ell}(i)=j, \mathbb{K}_i \right] \notag \\
		&~~=\mathbb{E} \left[J_{k}^\text{act}(\bm{a}_{\ell k}(t)=0) \Big| \bm{w}_{k,t-1}\right] -B(\bm{a}_{\ell k}(t)=1)
		\bm{b}_k(i) \notag
	\end{align}
	Let us denote by $\bm{\theta}_{k \ell}^j(t)$ the value of $\bm{\theta}_{k \ell}(t)$ at time $t$ if $\bm{a}_{k \ell}(i)=j$ is selected at time $i$. We utilize the assumption $\bm{\theta}_{\ell k}(t) = \bm{\theta}_{\ell k}(i)$ to rewrite $B(\bm{a}_{\ell k}(t)=1)$ as
	\begin{align}
		B(\bm{a}_{\ell k}(t)=1) =\bm{\theta}_{k \ell}^j(t) \bm{\theta}_{\ell k}(i) 
	\end{align}
	It then follows that
	\begin{align}
		&\mathbb{E} \left[
		J_{k}^\text{act}(\bm{a}_{\ell k}(t)) \Big{|} 
		\bm{w}_{k,i-1}, \bm{a}_{k \ell}(i)=j, \mathbb{K}_i \right] \notag \\
		&~~~= \mathbb{E} \left[J_{k}^\text{act}(\bm{a}_{\ell k}(t)=0) \Big| \bm{w}_{k,t-1}\right] - \bm{\theta}_{k \ell}^j(t) \bm{\theta}_{\ell k}(i) \bm{b}_k(i) \notag
	\end{align}
	Therefore, using (\ref{deltaJ2}) we get 
	\begin{align}
		\label{deltaJ22}
		\triangle J^2_k(t) = [\bm{\theta}_{k \ell}^0(t)-\bm{\theta}_{k \ell}^1(t)] \bm{\theta}_{\ell k}(i) \bm{b}_k(i)
	\end{align}
	Now, we recall that following Assumption~\ref{ass:rationalAgents} and the considered scenario $\mathbf{1}_{k \ell}(i)=1$, we have $\mathbf{1}_{k \ell}(t)=\mathbf{1}_{k \ell}(i)=1$ for $t\geq i$. As a result, the reputation update in (\ref{repu2}) can be approximated by expression (\ref{repuo}) for sufficiently small $\varepsilon$. Then, under (\ref{futureaction}), the future reputation scores $\bm{\theta}_{k \ell}^0(t)$ and $\bm{\theta}_{k \ell}^1(t)$ are given by:
	\begin{align}
		\label{deltathe}
		\bm{\theta}_{k \ell}^0(t) &= \bm{\theta}_{k \ell}(i) r_k^{t-i} \\
		\bm{\theta}_{k \ell}^1(t) &= \bm{\theta}_{k \ell}(i) r_k^{t-i}+(1-r_k)\sum\limits_{q=0}^{t-i-1}r_k^q \notag \\
		&= \bm{\theta}_{k \ell}(i) r_k^{t-i}+(1-r_k^{t-i})
	\end{align}
	Therefore, expression (\ref{deltaJ22}) becomes
	\begin{align}
		\label{j2}
		\triangle J_{k}^2(t) &
		= -(1-r_k^{t-i}) \bm{\theta}_{\ell k}(i) \bm{b}_k(i)
		,~~\text{for~~} t>i
	\end{align}
	Using (\ref{j1}) and (\ref{j2}), agent $k$ then chooses $\bm{a}_{k \ell}(i)=1$ if
	\begin{align}
		\label{condition}
		&~~~~c_k+\sum_{t=i+1}^{\infty} \delta_k^{t-i} c_k-\sum_{t=i+1}^{\infty} \delta_k^{t-i}(1-r_k^{t-i}) \bm{\theta}_{\ell k}(i) \bm{b}_k(i)<0 \notag \\
		&\Leftrightarrow \sum_{t=i}^{\infty} \delta_k^{t-i} c_k< \sum_{t=i+1}^{\infty} \delta_k^{t-i}(1-r_k^{t-i}) \bm{\theta}_{\ell k}(i) \bm{b}_k(i) \notag \\
		&\Leftrightarrow \frac{c_k}{1-\delta_k} < \bm{\theta}_{\ell k}(i) \bm{b}_k(i) \delta_k \cdot  \left(\frac{1}{1-\delta_k}-\frac{r_k}{1-\delta_k r_k}\right) \notag \\
		&\Leftrightarrow \bm{\gamma}_k(i)\triangleq \frac{\bm{b}_k(i)}{c_k}> \frac{\chi_k}{\bm{\theta}_{\ell k}(i)}
	\end{align}
	where we introduced 
	\begin{align}
		\chi_k& \triangleq \frac{1- \delta_k r_k }{\delta_k(1-r_k)}
	\end{align} 
	The best response rule $f_k(\cdot)$ therefore becomes
	\begin{align}
		a_{k \ell}(i)=
		\begin{cases}
			1, & \text{if}~ \bm{\gamma}_k(i)\triangleq \frac{\bm{b}_k(i)}{c_k}> \frac{\chi_k}{\bm{\theta}_{\ell k}(i)}\\
			0, & \text{otherwise} 
		\end{cases}
	\end{align}		
			
	\bibliographystyle{IEEEtran}
	\bibliography{IEEEabrv,refs}
	
	\begin{IEEEbiography}[{\includegraphics[width=1in,height=1.25in,clip,keepaspectratio]{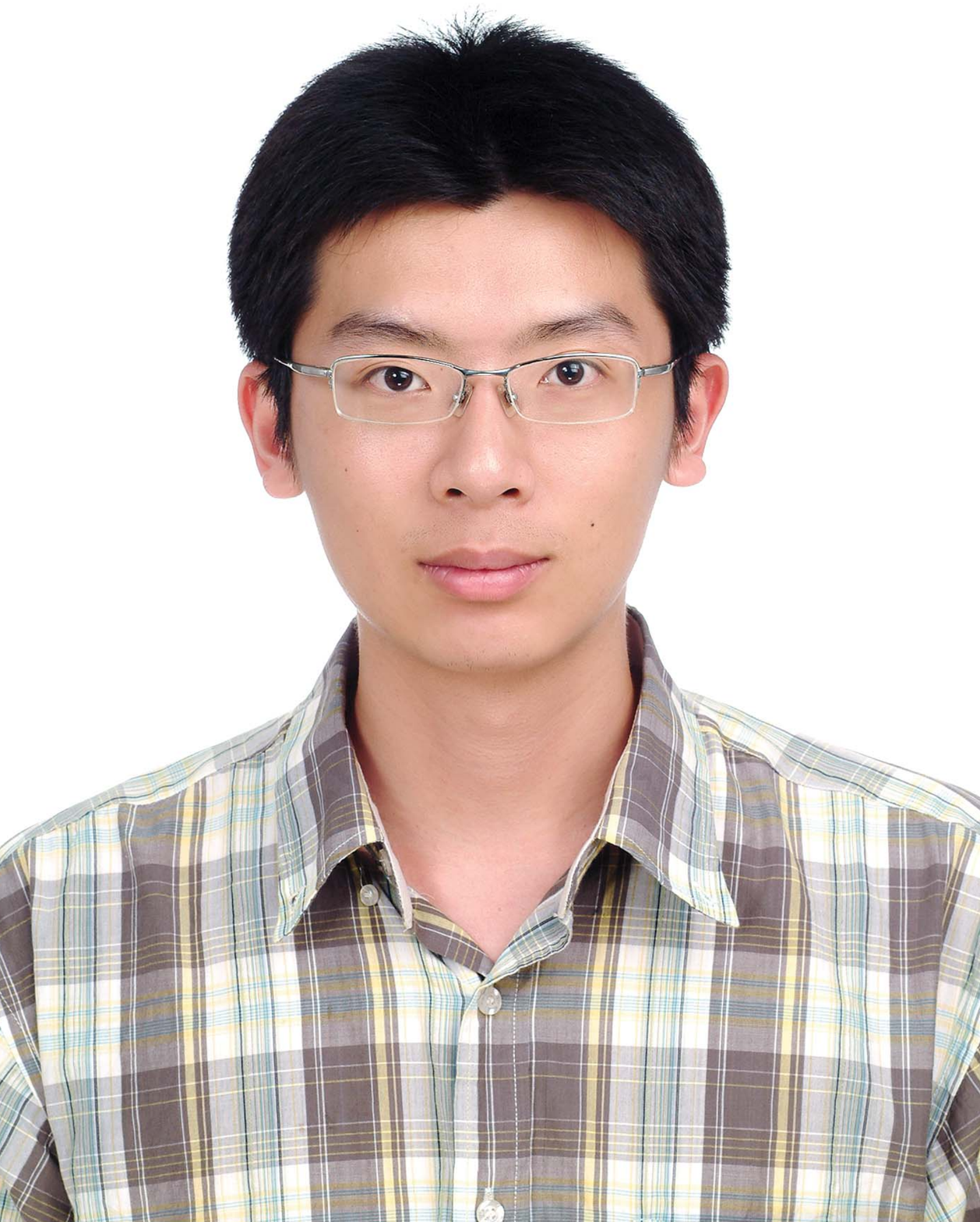}}]{Chung-Kai Yu}
		(S'08) received the B.S. and MS degrees in electrical engineering from the National Taiwan University (NTU), Taiwan, in 2006 and 2008, respectively.
		
		From 2008 to 2011, he was a Research Assistant in the Wireless Broadband Communication System Laboratory at NTU. He is currently working toward the Ph.D. degree in electrical engineering at the University of California, Los Angeles. His research interests include game-theoretic learning, adaptive networks, and distributed optimization.
	\end{IEEEbiography}
	
	\begin{IEEEbiography}[{\includegraphics[width=1in,height=1.25in,clip,keepaspectratio]{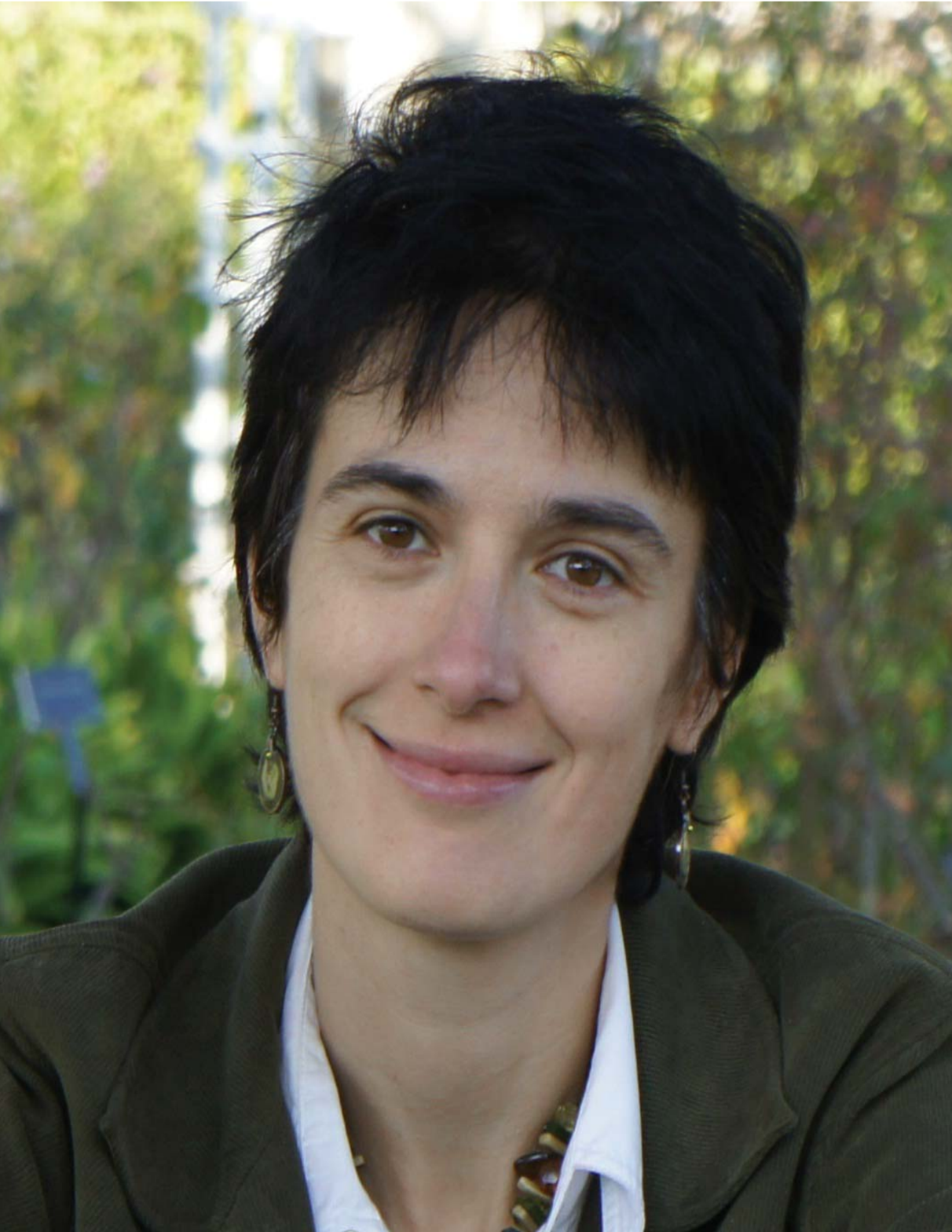}}]{Mihaela van der Schaar}
		(F'10) is Chancellor's Professor of Electrical Engineering at University of California, Los Angeles. 
		She is an IEEE Fellow since 2010, a Distinguished Lecturer of the Communications Society for 2011-2012, the Editor in Chief of IEEE Transactions on Multimedia (2011 - 2013).  She received an NSF CAREER Award (2004), the Best Paper Award from IEEE Transactions on Circuits and Systems for Video Technology (2005), the Okawa Foundation Award (2006), the IBM Faculty Award (2005, 2007, 2008), the Most Cited Paper Award from EURASIP: Image Communications Journal (2006), theGamenets Conference Best Paper Award (2011) and the 2011 IEEE Circuits and Systems Society Darlington Award Best Paper Award. She holds 33 granted US patents. She is also the founding and managing director of the UCLA Center for Engineering Economics, Learning, and Networks (see http://netecon.ee.ucla.edu). For more information about her research visit:  http://medianetlab.ee.ucla.edu/
		
	\end{IEEEbiography}
	
	\begin{IEEEbiography}[{\includegraphics[width=1in,height=1.25in,clip,keepaspectratio]{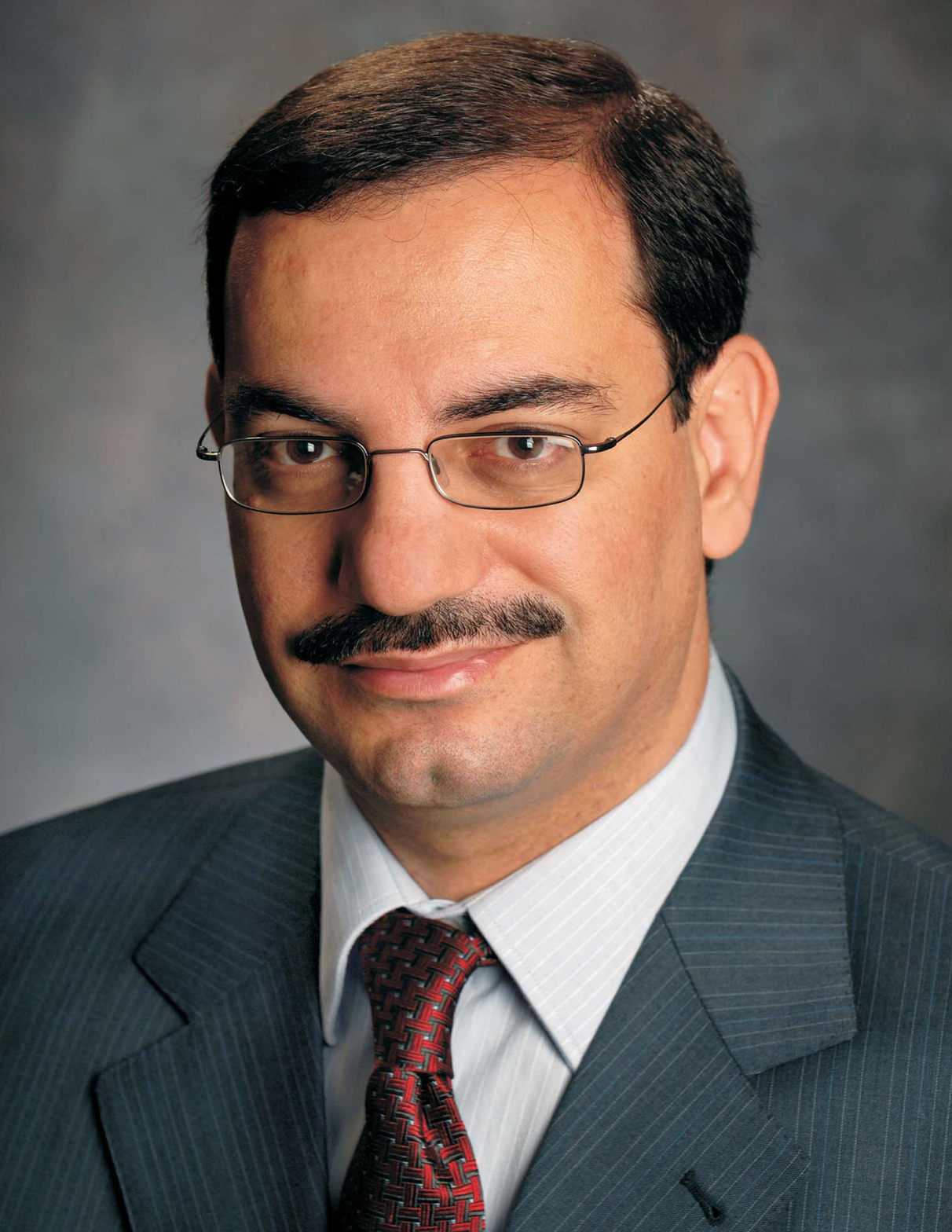}}]{Ali H. Sayed}
		(S'90-M'92-SM'99-F'01) is professor and former chairman of electrical engineering at the University of California, Los Angeles, USA, where he directs the UCLA Adaptive Systems Laboratory. An author of more than 440 scholarly publications and six books, his research involves several areas including adaptation and learning, statistical signal processing, distributed processing, network science, and biologically inspired designs.
		
		Dr. Sayed has received several awards including the 2014 Athanasios Papoulis Award from the European Association for Signal Processing, the 2013 Meritorious Service Award, and the 2012 Technical Achievement Award from the IEEE Signal Processing Society. Also, the 2005 Terman Award from the American Society for Engineering Education, the 2003 Kuwait Prize, and the 1996 IEEE Donald G. Fink Prize. He served as Distinguished Lecturer for the IEEE Signal Processing Society in 2005 and as Editor-in Chief of the IEEE TRANSACTIONS ON SIGNAL PROCESSING (2003–2005). His articles received several Best Paper Awards from the IEEE Signal Processing Society (2002, 2005, 2012, 2014). He is a Fellow of the American Association for the Advancement of Science (AAAS). He is recognized as a Highly Cited Researcher by Thomson Reuters.		
	\end{IEEEbiography}

\end{document}